\documentclass[a4paper,10pt]{article}
\pdfoutput=1
\usepackage{latexsym}
\usepackage{amsmath}
\usepackage{amsfonts}
\usepackage{amssymb}
\usepackage{graphicx}
\usepackage{amsthm}
\usepackage{a4wide}
\usepackage{enumerate}
\usepackage{url}

\usepackage[ruled,vlined]{algorithm2e}

\usepackage{epsfig}
\usepackage{figlatex}
\usepackage{setspace}
\graphicspath{{Figures/}}

\newtheorem{fact}{Fact}[section]
\newtheorem{lemma}[fact]{Lemma}
\newtheorem{theorem}[fact]{Theorem}
\newtheorem{corollary}[fact]{Corollary}
\newtheorem{proposition}[fact]{Proposition}

\newtheorem{definition}[fact]{Definition}

\theoremstyle{definition}

\newtheorem{ex}{Example}
\newtheorem{rem}{Remark}

\newcommand\R{\mathbb R}
\newcommand\N{\mathbb N}
\newcommand\Z{\mathbb Z}
\newcommand\Zp{\mathbb Z_+^*}
\newcommand\Zm{\mathbb Z_-^*}

\def\eref#1{(\ref{#1})}

\newcommand{\indep}{ {\perp\!\!\!\perp} }

\def\CE{\mathcal{E}}
\def\cF{\mathcal{F}}
\def\cN{\mathcal{N}}
\def\cA{\mathcal{A}}
\newcommand{\expct}[1]{\ensuremath{\text{{E}$\left[#1\right]$}}}

\def\Cb{C_{\bullet}}
\def\Cr{C_{\circ}}
\def\Cz{C_{\circledcirc}}
\def\Sb{S_{\bullet}}
\def\Sr{S_{\circ}}
\def\Sz{S_{\circledcirc}}

{\end{list}}%

\newenvironment{itemi}
{%
  \begin{list}{$\bullet$}%
  {\noindent%
    \usecounter{enumi}%
    \setlength{\topsep}{2pt}%
    \setlength{\partopsep}{0pt}%
				\setlength{\itemsep}{2pt}%
    \setlength{\parsep}{0pt}%
    \setlength{\leftmargin}{2.5em}%
    \setlength{\labelwidth}{1.5em}%
    \setlength{\labelsep}{0.5em}%
    \setlength{\listparindent}{0pt}%
    \setlength{\itemindent}{0pt}%
  }%
}%
{\end{list}}%

\begin{document}

\title{Stability of the bipartite matching model}
\author{Ana {\sc Bu\v si\'c}\thanks{INRIA/ENS, 23, avenue d'Italie, CS 81321, 75214 Paris Cedex 13, France. 
E-mail: {\tt ana.busic@inria.fr}.} \and Varun {\sc
    Gupta}\thanks{Computer Science Department, Carnegie Mellon University, Pittsburgh, PA, USA. 
E-mail: {\tt varun@cs.cmu.edu}.} \and Jean {\sc Mairesse}
\thanks{LIAFA, CNRS et Univ. Paris 7, case 7014,
75205 Paris Cedex 13, France.
E-mail: {\tt mairesse@liafa.jussieu.fr.}}}


\maketitle

\begin{abstract}
We consider the bipartite matching model of customers and servers introduced by Caldentey, 
Kaplan, and Weiss (Adv. Appl. Probab., 2009). Customers
and servers play symmetrical roles. There is a finite set $C$,
resp. $S$, of customer, resp. server, 
classes. 
Time is discrete and at each time step, one customer and one server arrive in the
system according to 
a joint probability measure $\mu$ on $C\times S$, independently of the
past.  Also, at each time step, pairs of {\em matched}
customer and server, if they exist, depart from the
system. Authorized {\em matchings} are given by a fixed bipartite
graph  $(C, S, E\subset C \times S)$. 
A {\em matching policy} is chosen,
which decides how to match when there are several possibilities. 
Customers/servers that cannot be matched are stored in a
buffer.  
The evolution of the
model can be described by a discrete time Markov chain. We study
its stability under various  admissible matching 
policies including: ML (Match the Longest),
MS (Match the Shortest), FIFO (match the oldest), priorities.
There exist natural necessary conditions for stability (independent of the matching
policy) defining the maximal possible stability region. For some bipartite graphs, we prove that the
stability region
is indeed maximal for any
admissible matching policy. For the ML policy, we prove that the stability region
is maximal for any bipartite graph. For the MS and priority policies,
we exhibit a bipartite graph with a non-maximal stability
region.
\end{abstract}

\smallskip


\smallskip

{\noindent\bf Mathematics Subject Classification (MSC2010):} 60J10,
60K25, 68M20, 05C21. 

\tableofcontents


\section{Introduction}
In queueing theory, customers and servers play
different roles. Customers arrive in the system, accumulate in a
buffer, get served by a server, and eventually depart. Servers on the other hand
alternate between idle and busy periods but remain forever in the
system. 

Within this framework, many variations and refinements are
possible. For instance, we may consider a model with multi-class 
customers and distinguishable servers. A customer of a given class $c$
must choose its server from a specified subset $S(c)$ of the
servers. And of course, the subsets $S(c)$ may intersect, see Figure
\ref{fig:call}.  

\begin{figure}[htb]
  \centering
  \includegraphics[width=.28\textwidth]{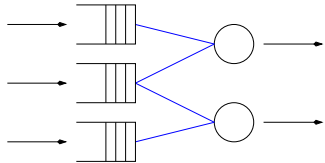}
  \caption{\label{fig:call}Queueing model of a call center.}
\end{figure}

\medskip

In this paper, we consider a model with the same multi-class flavor, but
in which, by contrast, customers and servers
play completely identical roles. We now argue that this simple
symmetry requirement leads in a natural and ineluctable way 
to the {\em bipartite matching model}. 

\medskip

By symmetry, both customers and
servers should arrive into the system and depart from it. More specifically, upon
completion of a service, both the customer and the server should depart
simultaneously. 
To model arrivals, we have a priori more flexibility, but there is
basically one non-trivial choice which is to assume that time is
discrete and that customers and servers arrive in pairs. 

Consider indeed 
the simplest possible model with continuous-time arrivals: (i) there is only one class of customers and one
class of servers; (ii) customers, resp. servers, arrive according to a
Poisson process of rate $\lambda$, resp. $\mu$; (iii) services have
duration 0. 
Let us describe the state by $Z=X-Y$, where $X$ is the number of
unmatched customers and $Y$ the number of unmatched servers. The
process $Z$ is a birth-and-death continuous-time Markov process on
$\Z$ with drift $\lambda-\mu$. It is either transient (if $\lambda\neq
\mu$) or null recurrent (if $\lambda=\mu$), but it is never positive
recurrent. 

Let us switch to discrete-time i.i.d. arrivals. At each time step,  
a batch of customers and a batch of servers arrive into the system. If
the size of the batches are allowed to be different for customers
and servers, then we are back to the continuous-time situation, and
even the simplest model is never positive recurrent. 

Therefore to get a non-trivial model, the natural assumption is that exactly one customer
and one server arrive into the system at each time step. The resulting
model is symmetric in another respect: both arrivals and departures
occur in pairs. 

For simplicity, we always
assume that the 
service durations are null. So the model is specified by: (i) the
finite set $C$ of customer classes and the finite set $S$ of server
classes; (ii) the probability law $\mu$ on $C\times S$ for the arrivals in pairs; (iii) the
bipartite graph $(C, S, E\subset C \times S)$ giving the possible
matchings between customers and servers (hence the possible
departures in pairs); (iv) the matching policy to decide how to match when
several choices are possible. We consider so called {\em admissible}
policies which depend only on the current state of the system. Under
these assumptions, the buffer content evolves as a discrete-time Markov chain. 

We call this model the {\em bipartite matching} model. 

\medskip

The bipartite matching model has been introduced by Caldentey, 
Kaplan, and Weiss~\cite{CKW09}, under an additional assumption of
independence between arriving customers and servers ($\forall c,s,
\ \mu(c,s)=\mu(c,S)\mu(C,s)$), and for the FIFO policy.  
In their paper, the authors mention several
possible domains of applications ranging from call centers to crossbar
data switches. They also provide references to papers on
related models. We refer the interested reader to
\cite{CKW09} for details. 

\medskip

In the bipartite matching model, there is an equal 
number of customers and servers at any time. But the matching
constraints may result in instability with unmatched
customers and servers accumulating. 
It turns out that proving stability, i.e. positive recurrence of
the Markov chain, is highly non-trivial. 
Given a bipartite graph $(C,S,E)$, there exist natural necessary
conditions on $\mu$ for stability to hold true. When these conditions
are also sufficient, we say that the stability region is {\em
  maximal}. 

\medskip

Caldentey $\&$ al conjecture that any 
bipartite graph has a maximal stability region for the FIFO
policy~\cite[Conjecture 4.2]{CKW09}. 
They prove the conjecture for some specific models (under the additional assumption of independence of
customers and servers): (i) the N model defined by $C=\{1,2\},S=\{1',2'\}, E=
  C\times S - \{(2,2')\}$, (ii) the W model, i.e. the matching model
  version of Figure \ref{fig:call}; (iii) the NN model of Figure
  \ref{fig:NN}. In the first two cases, they are also able to compute
  explicitly the stationary distribution. For the last case, the proof
  of stability is already intricate. 

\medskip

In the present paper, we consider the stability issue for various
admissible matching policies: ML (Match the Longest),
MS (Match the Shortest), FIFO (match the oldest), random (match
uniformly), priorities. The irreducibility of the Markov chain
describing the model is not granted, and we first study this question
in detail (Section \ref{se-utc}). 
Then we obtain the following results:
\begin{itemi}
 \item[-] sufficient conditions under which any admissible policy is
   stable (Section \ref{se-suff}); 
 \item[-] for the NN model, the MS policy and some priority 
       policies do not have a maximal stability region (Section \ref{se-ms}); 
 \item[-] for any bipartite graph, the ML policy has a maximal stability
   region (Section \ref{se-ml}).   
\end{itemi}
We do not know if the stability region is
always maximal for the FIFO and random policies.

\medskip

{\bf Notations.} Denote by $\N = \{0, 1, 2, \ldots\}$ the set of
non-negative integers. 
Let $A^*$ be the free monoid generated by $A$. For any word $w \in A^*$ and
any  $B \subset A$, set $|w|_B = \# \{i \mid w_i \in B\}$, the number
of occurrences in $w$ of letters from $B$. For $B = \{b\}$, we shorten
the notation to $|w|_b$. Furthermore, 
for any $w \in A^*$, set $[w]:=(|w|_a)_{a\in A}$ (the commutative
image of $w$).

\section{The bipartite matching model}

We now proceed to a more formal definition of the model. 

\begin{definition}\label{de-structure}
A {\em bipartite matching structure} is a quadruple $(C,S,E,F)$ where 

\begin{itemi}
\item $C$ is the non-empty and finite set of customer types;
\item $S$ is the non-empty and finite set of server types;
\item $E \subset C\times S$ is the set of possible matchings;
\item $F \subset C\times S$ is the set of possible arrivals.
\end{itemi}

The bipartite graph $(C,S,E)$ is called the {\em matching graph}. 
It is assumed to be connected. The bipartite graph $(C,S,F)$
is called the {\em arrival graph}. It is assumed to have no isolated
vertices. 
\end{definition}

The two assumptions in Def.
\ref{de-structure} are made without loss of generality, see Remarks
\ref{re-wlog1} and \ref{re-wlog2}. 

\medskip

In Figure \ref{fig:NN} we give an example of a  matching graph with $3$ customer and $3$ server types, 
called the ``NN graph'' in the following.
\begin{figure}[htb]
  \centering
  \includegraphics[width=.25\textwidth]{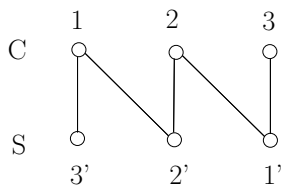}
  \caption{\label{fig:NN}NN graph.}
\end{figure}

Customers and servers play symmetrical roles in the model. Also $E$ and
$F$ play dual roles. 
The graph $(C,S,E)$ defines the pairs that may depart from the system,
while the graph $(C,S,F)$ defines the pairs that may arrive into the
system. 

\begin{definition}\label{de-model}
A {\em bipartite matching model} is a triple $[ (C,S,E,F), \mu,
  \mbox{{\sc Pol}}]$, where 

\begin{itemi}
\item $(C,S,E,F)$ is a bipartite matching structure;
\item $\mu$ is a probability measure on $C\times S$ satisfying 
\begin{equation}\label{eq-3cond}
\mbox{supp}(\mu) = F, \ \mbox{supp}(\mu_C)=C,
\ \mbox{supp}(\mu_S)=S \:,
\end{equation}
where $\mu_C$ and $\mu_S$ are the $C$ and $S$ marginals of $\mu$. 
\item {\sc Pol} is an admissible matching policy (to be defined in \S \ref{sse-match}).
\end{itemi}
\end{definition}

Observe that we can simplify the notation to $[ (C,S,E), \mu, \mbox{{\sc Pol}}]$.
We say that the model $[ (C,S,E), \mu,
  \mbox{{\sc Pol}}]$ is {\em associated} with the structure
$(C,S,E,F)$. 

\begin{rem}\label{re-wlog1}
For \eref{eq-3cond} to have solutions, $(C,S,F)$ must be
without isolated vertices, the assumption made in Definition
\ref{de-structure}. This is not a real restriction: if it is not
satisfied, we can consider a new model without such customer or server
classes. 
\end{rem}

A realization of the model is as follows. Consider an i.i.d. sequence
of random variables of law $\mu$, representing the arrival stream of
pairs of customer/server. 
A state of the buffer consists of an equal
number of customers and servers with no possible matchings between the
classes. Upon arrival of a new ordered pair $(c,s)$, two situations
may occur: if neither $c$ nor $s$ match with the servers/customers
already present in the buffer, then $c$ and $s$ are simply added to
the buffer; if $c$, resp. $s$, can be matched then it departs
the buffer with its match. If several matchings are possible for $c$,
resp. s, then it is the role of the matching policy to select one. 
An admissible policy selects according to the current state of the
buffer (and not according to the whole history of the buffer contents,
for instance). 
The resulting evolution of the buffer is described by a  discrete-time Markov
chain. 

\subsection{State space description}

Depending on the matching policy, we consider either a commutative
(e.g. for Random)
or a non-commutative (e.g. for FIFO) state space description. 
The different policies considered in the paper will be formally defined in
\S \ref{sse-match}. 

\medskip

Let us choose a matching graph $(C,S,E)$. We introduce the following convenient notations: $C(s)$ is the set of
customer classes that can be matched with an $s$-server; $S(c)$ is the
set of server classes that can be matched with a $c$-customer: 
$$S(c) = \{s \in S \; : \; (c,s) \in E\}, \quad C(s) = \{c \in C \; : \; (c,s) \in E\}.$$
For any subsets $A \subset C$, and $B \subset S$, we define 
$$S(A) = \cup_{c\in A}S(c), \quad C(B) = \cup_{s\in B}C(s).$$

\paragraph{Commutative state space.} 
A state of the system is given by $(x, y)$,  
$x = (x_c)_{c\in C}$ and $y = (y_s)_{s\in S}$, where $x_c$ denotes the number of customers of type $c$ and $y_s$ the number 
of servers of type $s$. The {\em commutative state space} is:
\begin{equation}\label{eq-css}
\CE = \Bigl\{ (x, y) \in \N^C\times \N^S \; : \; \sum_{c\in C} x_c = \sum_{s \in S} y_s; \; 
\forall  (c, s) \in E, \; x_cy_s = 0 \Bigr\}.
\end{equation}

\paragraph{Non-commutative state space.} A state of the system is given by two finite words 
of the same size $k \geq 0$, respectively on the alphabets $C$ and $S$,
describing unmatched customers and servers.
The {\em non-commutative state space} is:
\begin{equation}\label{eq-ncss}
\CE = \Bigl\{ (u,v) \in \cup_{k\geq 0} (C^{k}\times S^k) \; : \; ([u],
    [v]) \mbox{ belongs to } \eref{eq-css}  \Bigr\}.
\end{equation}

\paragraph{Facet.} Both the commutative and the non-commutative state space can be decomposed into facets, defined only by the 
non-zero classes. 

\begin{definition}
 A {\em facet} is an ordered pair $(U, V)$ such that:
 $U \subset C, V \subset S$ and $U \times
 V \subset (C \times S - E)$. The {\em zero-facet} is the
 facet $(\emptyset,\emptyset)$, we denote it shortly by $\emptyset$.
\end{definition}

For a facet $\cF = (U, V)$, define:
\begin{align*}
&\Cb(\cF)= U, & \qquad & \Sb(\cF) = V,  \\
&\Cz(\cF)= C(V), & \qquad & \Sz(\cF)= S(U),  \\
&\Cr(\cF) = C- (\Cb(\cF) \cup \Cz(\cF) ), & \qquad & \Sr(\cF) = S- (\Sb(\cF) \cup
\Sz(\cF) ).
\end{align*} 

We alleviate the notations to $\Cb, \Sb,\Cz, \dots, $ when there is no
possible confusion. 
The symbol $\bullet$ stands for the non-zero
classes, the symbol $\circledcirc$ for the classes that are forced to be at zero (since
they are matched with non-zero classes), and the symbol $\circ$ for
the classes that happen to be at zero. 

\medskip

The following notion will play an important role later on. 

\begin{definition}\label{de-saturated}
A facet $\cF$ is called saturated if $\Cr(\cF)= \emptyset$ or $\Sr(\cF) = \emptyset$. 
\end{definition}

In Figure \ref{fig:NNf33}, the facet on the left is non-saturated, while the
one on the right is saturated. 

\begin{figure}[htb]
  \centering
  \includegraphics[width=.5\textwidth]{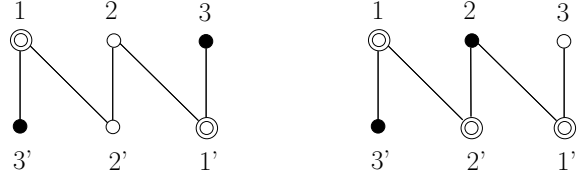}
  \caption{\label{fig:NNf33}NN graph: facets  $(\{3\}, \{3'\})$ and $(\{2\}, \{3'\})$.}
\end{figure}

\paragraph{Graphical convention.}

A facet $\cF$ 
can be represented graphically by
coloring the nodes of the bipartite graph  according to the above
convention (see Figure \ref{fig:NNf33} for an illustration): 

\begin{itemi}
 \item[-] nodes in $\Cb(\cF)$ and  $ \Sb(\cF)$ are represented as filled circles;
 \item[-] nodes in $\Cz(\cF)$ and  $\Sz(\cF)$ are represented as double circles;
 \item[-] nodes in $\Cr(\cF)$ and  $\Sr(\cF)$ are represented as simple circles.
\end{itemi}

In Figure \ref{fig:NNfacets}, we have represented the 
facets of the NN graph. The more complex case of the NNN graph
will be given in Section \ref{se-suff}, Figure \ref{fig:NNNfacets}.

\begin{figure}[htb]
  \centering
  \includegraphics[width=.40\textwidth]{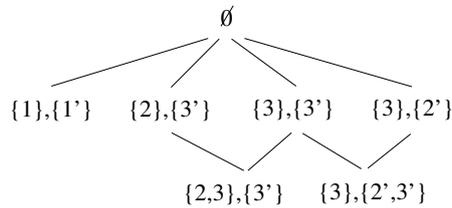}
  \caption{\label{fig:NNfacets}Facets for the NN graph.}
\end{figure}

Algorithm~\ref{al:facets} takes as input a matching graph and returns
as output the set of facets. The termination and correctness of the
algorithm are easily proved.

\begin{algorithm}
\KwData{
A bipartite graph $G=(C, S, E)$. 
}
\KwResult{$\mathit{Facets}$ - set of all facets of $G$.}
\Begin{
$\mathit{Facets} \leftarrow \emptyset$; $\mathit{New} \leftarrow \emptyset$\;
\lForEach{$(i,j)\in C\times S - E$}{
$\mathit{New} \leftarrow \mathit{New} \cup \{(\{i\}, \{j\})\}$\; 
}
\While{$\mathit{New} \not=\emptyset$}{
$\mathit{Facets} \leftarrow \mathit{Facets} \cup \mathit{New}$\;
$\mathit{Old} \leftarrow \mathit{New}$; $\mathit{New} \leftarrow \emptyset$\;
\ForAll{${\cal H}, {\cal K} \in \mathit{Old}$ such that ${\cal H} \not = {\cal K}$} {
\If{$\Cb({\cal H}) = \Cb({\cal K})$ or $\Sb({\cal H}) = \Sb({\cal K})$}{
  $Z \leftarrow (\Cb({\cal H}) \cup \Cb({\cal K}), \; \Sb({\cal H}) \cup \Sb({\cal K}))$\;
  $\mathit{New} \leftarrow \mathit{New} \cup \{Z\}$\;
}
}
}
$\mathit{Facets} \leftarrow \mathit{Facets} \cup \{\emptyset\}$\;
\Return{$\mathit{Facets}$\;}}
\caption{Computation of the  facets}
\label{al:facets}
\end{algorithm}

\subsection{Admissible matching policies}\label{sse-match}

Informally, a matching policy is {\em admissible} if:
\begin{itemi}
\item[-] only the current state of the buffer is taken into account;
\item[-] priority is given to customers/servers 
that are already present in the buffer: if the state is $(u,v)$ and
the new arrival is $(c,s)\in E$, then $c$ and $s$ are matched together
iff there are no servers from $S(c)$ in $v$ and no customers from
$C(s)$ in $u$. 
\end{itemi} 

It results from the first point that an admissible matching policy can
be described as a mapping $\odot : \CE \times (C\times S) \rightarrow
\CE$ which returns the new state of the system after an arrival. 
The second point is called the {\em buffer-first} assumption. It is not a real restriction:
a matching policy that always gives priority to new arrivals can be seen as a special case 
of the above with an arrival probability  $\mu$ such that $\mu(E) = 0$. 

\medskip

We now define admissible policies formally, distinguishing between the
non-commutative and commutative state spaces.

For a word $w \in A^k$ and $i\in \{1,\dots, k\}$, we denote by 
$w_{[i]}:= w_1\ldots w_{i-1} w_{i+1}\ldots w_k$ the
subword of $w$ obtained by deleting $w_i$. 

\begin{definition}[Non-commutative case]
\label{def:bf}
A matching policy is {\em admissible} if there are functions $\Phi$ and $\Psi$ such that:
$$
(u, v) \odot (c,s) = 
\left \{
\begin{array}{ll}
(uc, vs), & \textrm{if } \; |u|_{C(s)} = 0, \;|v|_{S(c)} = 0, \;
  (c,s) \not\in E \\
(u, v), & \textrm{if } \; |u|_{C(s)} = 0, \; |v|_{S(c)} = 0, \; (c,s)\in E\\ 
(u_{\left [\Phi(u,s)\right]}, v_{\left [\Psi(v,c)\right]}), & \textrm{if } \; |u|_{C(s)} \not= 0, \; |v|_{S(c)} \not= 0\\ 
(u_{\left [\Phi(u,s)\right]}c, v), & \textrm{if } \; |u|_{C(s)} \not= 0, \; |v|_{S(c)} = 0\\ 
(u, v_{\left [\Psi(v,c)\right]}s), & \textrm{if } \; |u|_{C(s)} = 0, \; |v|_{S(c)} \not= 0
\end{array}
\right .
$$
\end{definition}

The {\em FIFO} and {\em LIFO} policies are admissible matching policies with functions $\Phi$ and $\Psi$ as follows:
\begin{itemize}
 \item FIFO : $\Phi(u, s) = \arg\min \{u_k \in C(s)\}$, 
$\Psi(v, c) = \arg\min \{v_k \in S(c)\}$.
 \item LIFO : $\Phi(u, s) = \arg\max \{u_k \in C(s)\}$, 
$\Psi(v, c) = \arg\max \{v_k \in S(c)\}$.
\end{itemize}

For $c\in C$, let $e_c\in \N^C$ be defined by $(e_c)_c=1$ and
$(e_c)_d=0, d\neq c$. For $s\in S$, let $e_s$ be defined accordingly. 

\begin{definition}[Commutative case]
 \label{def:bfc}
A matching policy is {\em admissible} if there are functions $\Phi$ and $\Psi$ such that:
$$ 
(x, y) \odot (c,s) = 
\left \{
\begin{array}{ll}
(x+e_c, y+e_s), & \textrm{if } \; x_{C(s)} = 0, \; y_{S(c)} = 0, \;
  (c,s) \not\in E\\
(x, y), & \textrm{if } \; x_{C(s)} = 0, \; y_{S(c)} = 0, \; (c,s)\in E\\ 
(x - e_{\Phi(x,s)}, y - e_{\Psi(y,c)}), & \textrm{if } \; x_{C(s)} \not= 0, \; y_{S(c)} \not= 0\\ 
(x - e_{\Phi(x,s)} + e_{c}, y), & \textrm{if } \; x_{C(s)} \not= 0, \; y_{S(c)} = 0\\ 
(x,  y - e_{\Psi(y,c)} + e_s), & \textrm{if } \; x_{C(s)} = 0, \; y_{S(c)} \not= 0
\end{array}
\right .
$$
\end{definition}

The following commutative matching policies are admissible (for RANDOM, ML, and MS policies 
$\Phi(u, s)$ and $\Psi(v, c)$ are random variables):

\begin{itemize}

\item PR (Priorities). For each customer type $c\in C$, we define a priority function 
$\alpha_c : S(c) \rightarrow \{1, \ldots, |S(c)|\}$.
Similarly, for each server type $s \in S$, we define $\beta_s : C(s) \rightarrow \{1, \ldots, |C(s)|\}$.
In the case of several matching options, a customer/server is matched with the server/customer that has the 
highest priority (greatest value of the priority function).  
It is convenient to specify the priorities by two  $|C|\times |S|$
matrices $A$  and $B$ defined by:
$$
A_{cs} = \left \{ 
\begin{array}{ll}
 \alpha_c(s), & (c,s) \in E \\
 0, & \textrm{otherwise}
\end{array}
\right .
 \quad 
B_{cs} = \left \{ 
\begin{array}{ll}
\beta_s(c), & (c,s) \in E \\
 0, & \textrm{otherwise}
\end{array}
\right . \:.
$$
Then $\Phi(x, s) = \arg\max \{ \beta_s(c) \; : \; c \in C(s), x_c>0\}$ and 
$\Psi(y, c) = \arg\max \{ \alpha_c(s) \; : \; s \in S(c), y_s>0\}$. 

\item RANDOM : $\Phi(x, s)$, resp. $\Psi(y, c)$, is a random variable
valued in $C(s)$, resp. $S(c)$, and distributed as $\left
(x_i/\sum_{j \in C(s)}x_j\right )_{i \in C(s)}$, 
resp. $\left (y_i/\sum_{j \in S(c)}y_j\right )_{i \in S(c)}$. Intuitively, the
match is chosen uniformly among all possible ones. 

\item ML : $\Phi(x, s)$, resp. $\Psi(y, c)$, is a  random variable
  uniformly distributed on 
$\arg \max \{ x_i : i \in C(s) \}$, 
resp. $\arg \max \{ y_i : i \in S(c) \}$.

\item MS : $\Phi(x, s)$, resp. $\Psi(y, c)$, is a random variable
  uniformly distributed on
$\arg \min \{ x_i>0 : i \in C(s) \}$, resp. $\arg \min \{ y_i>0 : i \in S(c) \}$.
\end{itemize}

\section{Necessary conditions for stability}

To introduce the main ideas, consider first a simpler finite and deterministic
problem. Let $(C,S,E)$ be a matching graph. Consider a batch of
customers $x\in \N^C$ and
a batch of servers $y\in \N^S$ of equal size: $\sum_c x_c = \sum_s
y_s$. 
A {\em perfect matching} of $x$ and $y$ is a tuple $m\in \N^E$ such that: 
\[
\forall c\in C, \ x_c = \sum_{s \in S(c)} m_{cs}, \qquad
\forall s \in S, \ y_s = \sum_{c \in C(s)} m_{cs} \:.  
\]
By Hall's Theorem (aka the ``marriage Theorem''), there exists a
perfect matching if and only if: 
\begin{equation}\label{eq-hall}
\begin{array}{ll}
\sum_{c\in U} x_c \leq \sum_{s\in S(U)} y_s, \quad \forall U  \subset C\\
\sum_{s\in V} y_s \leq \sum_{c\in C(V)} x_c, \quad \forall V  \subset S
\end{array}
\end{equation}
A perfect matching, if there is one, can be obtained by restating the
model as a flow network and by solving the maximum flow problem for
which efficient algorithms exist~\cite{FoFu,EdKa}. 

\medskip

The bipartite matching model is much more complicated: first it is random,
and second the matchings
have to be performed on the fly, at each time step. 
However the two ingredients of the simpler model will play an
instrumental role in the analysis: (i) the conditions {\sc NCond}, to
be defined in \eref{eq:CNat}, are related to \eref{eq-hall}; (ii) the
restatement as a flow problem is used in most of the proofs.

\medskip

Consider now a bipartite matching model $[ (C,S,E), \mu, \mbox{{\sc Pol}}]$.
We identify the model with the Markov chain on the state space $\CE$ describing the evolution
of the buffer content. 

\medskip

Let $P$ be the transition matrix of the Markov chain. A probability
measure $\pi$ on $\CE$ is {\em stationary} if $\pi P=\pi$. It is {\em
  attractive} if for any probability measure $\nu$ on $\CE$, the
sequence of Cesaro averages of $\nu P^n$ converges weakly to $\pi$. 

\begin{definition}\label{de-stab}
The model is said to be {\em stable} if the Markov chain has a unique
and attractive 
stationary probability measure. 
\end{definition}

It implies in particular that the graph of the Markov chain has
a unique terminal strongly connected component with all states leading
to it. 

\medskip

Let $\mu_C$ be a probability measure on $C$ and $\mu_S$ a probability
measure on $S$. 
Define the following conditions on $(\mu_C,\mu_S)$:

\begin{equation}
\label{eq:CNat}
\textrm{{\sc NCond}}: \quad \left \{ 
\begin{array}{ll}
\mu_C(U) < \mu_S(S(U)), \quad \forall U  \subsetneq C\\
\mu_S(V) < \mu_C(C(V)), \quad \forall V  \subsetneq S
\end{array}
\right .
\end{equation}

\medskip

The above conditions appear in \cite{CKW09}. They have a natural interpretation.  Let $\mu_C$ and
$\mu_S$ be the marginals of the arrival probability $\mu$.
Customers from $U$
need to be matched with servers from $S(U)$. The first line in
{\sc NCond} asks for strictly more servers in average from $S(U)$ than customers from
$U$. The second line has a dual interpretation. 

Using the Strong Law of Large Numbers, we also see that the arrivals up to
time $n$ satisfy \eref{eq-hall} for all values of $n$ large enough if
and only if {\sc NCond} is satisfied. 

\begin{lemma}\label{le-le}
The conditions {\sc NCond} are necessary stability conditions: if the 
Markov chain is stable then the 
conditions {\sc NCond} are satisfied by the marginals of $\mu$.
\end{lemma}

\begin{proof}
We suppose that the conditions {\sc NCond}  are not satisfied. 

Assume first that there exists $U\subset C$ such that $\mu_C(U) >
\mu_S(S(U))$. 
Let $A_n$, resp. $B_n$, be the total numbers of customers of type $U$,
resp. servers of type $S(U)$, to arrive in the system up to time
$n$. Let $X_n$ be the number of customers of type $U$ present in the
system at time $n$. By definition, $X_n \geq A_n - B_n$. By the Strong
Law of Large Numbers, we have, a.s., 
\[
\lim_n \frac{A_n}{n} = \mu_C(U), \ \lim_n \frac{B_n}{n} = \mu_S(S(U)),
\quad \lim_n \frac{X_n}{n} \geq \mu_C(U) - \mu_S(S(U)) > 0 \:.
\]
So the Markov chain is transient. Similarly, if there exists $V\subset
S$ such that $\mu_S(V) > \mu_C(C(V))$, the model is unstable. (This
part of the argument appears in \cite[Prop. 3.4]{CKW09}.) 

Assume now that there exists $U\subset C, U \neq C,$ such that 
\begin{equation}\label{eq-1}
\mu_C(U) = \mu_S(S(U))\:.
\end{equation}
Observe that $S(U)\neq S$, otherwise we would have
$\mu_C(U)=\mu_S(S)=1$ which would contradict $U\neq C$.  
Set $V=S-S(U)$. Eqn (\ref{eq-1}) is equivalent to: $\mu_S(V) =
\mu_C(C-U)$. 
The bipartite matching graph $(C,S,E)$ is represented in Figure
\ref{fi-help}. By assumption, $U\times V \cap E = \emptyset$. 

 \begin{figure}[htb]
   \centering
   \includegraphics[width=.4\textwidth]{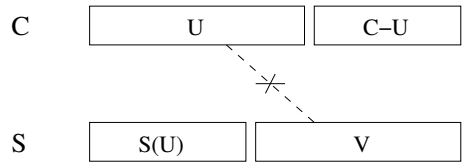}
   \caption{\label{fi-help}The bipartite graph $(C,S,E)$.}
 \end{figure}

We have
\begin{eqnarray*}
\mu(U\times V) & = & \mu_C(U) - \mu(U\times S(U)) \\
\mu((C-U)\times S(U)) & = & \mu_S(S(U)) - \mu(U\times S(U))\:.
\end{eqnarray*}
Using Eqn (\ref{eq-1}), we get
\begin{equation}\label{eq-2}
\mu(U\times V) = \mu ((C-U)\times S(U) ) \:.
\end{equation}

Let $D_n$ be the number of departures
of type $(C-U)\times S(U)$ up to time $n$. 
Let $A_n$, $B_n$, and $X_n$ be defined as above. 
Set $Z_n = A_n - B_n$. For an arrival of type
$U\times V$, the $Z$-process makes a +1 jump, for an arrival of type
$(C-U)\times S(U)$, the $Z$-process makes a -1 jump, otherwise the
$Z$-process remains constant. We have $X_n \geq A_n - (B_n - D_n) \geq Z_n$.
We have two cases:
\begin{itemi}
 \item[-] If $\mu(U\times V) >0$, then, according to \eref{eq-2}, the $Z$-process is null recurrent.
 \item[-] If $\mu(U\times V) =0$, then for any initial condition such that 
$X_0 > 0$, a.s. $X_n \geq X_0, \; \forall n$.
\end{itemi}
Hence, in both cases, the model cannot be stable.
\end{proof}

\begin{rem}\label{re-wlog2}
Consider a {\em non-connected} matching graph
$(C,S,E)$. Consider a probability
$\mu$ and an admissible matching policy  such that the bipartite
matching model is
stable. 
Let $(C',S',E')$ be a connected subgraph of $(C,S,E)$. 
Following the exact same steps as in the proof of Lemma \ref{le-le},
we prove that 
\[
\mu_C(C')=\mu_S(S'), \quad \mu(C' \times (S-S'))=0, \quad \mu((C-C') \times S')=0
\]
(otherwise the Markov chain is
either transient or null recurrent). Therefore, we can decompose 
the model into connected components and treat them separately. 
Hence the assumption of connectedness of $(C,S,E)$ in
Def. \ref{de-structure} was made
without loss of generality.
\end{rem}

\subsection{Complexity of verifying {\sc NCond}} 

Let us fix $(C,S,E)$ and the probability
measures $(\mu_C,\mu_S)$ such that $\mbox{supp}(\mu_C)=C,
\ \mbox{supp}(\mu_S)=S$. We want an efficient algorithm to decide 
if the conditions {\sc NCond} are satisfied. 

\medskip

The number of inequalities in {\sc NCond} is exponential in $|C| +
|S|$. So checking directly if all the inequalities are satisfied is a method
whose time complexity is exponential in $|C| +
|S|$. To go beyond, we need additional material. 

\medskip

We use the standard terminology of network flow theory, see for
instance \cite{FoFu}.
Consider the directed graph
\begin{equation}\label{eq-dg}
\cN= \bigl(C\cup S \cup \{i,f\}, E\cup\{(i,c), c\in
C\} \cup \{(s,f), s\in S\}\bigr)\:. 
\end{equation}
Endow the arcs of $E$ with infinite capacity, an arc of type $(i,c)$
with capacity $\mu_C(c)$, and an arc of type $(s,f)$ with capacity
$\mu_S(s)$. 

 \begin{figure}[htb]
   \centering
   \includegraphics[width=.2\textwidth]{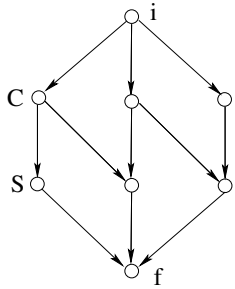}
   \caption{\label{fi:maxmin}The graph $\cN$ associated
with the NN model of Figure \ref{fig:NN}.}
 \end{figure}

Recall that a {\em  cut} is a subset of the arcs
whose removal disconnects $i$ and $f$. The {\em capacity} of a cut is
the sum of the capacities of the arcs. 
Set $A= E\cup\{(i,c), c\in
C\} \cup \{(s,f), s\in S\}$.
Recall that $T: A \rightarrow \R_+$ is a
{\em flow} if: (i) $\forall
c, T(i,c)=\sum_{s\in S(c)} T(c,s), \ \forall s, \sum_{c\in
  C(s)}T(c,s)=T(s,f)$; (ii) $\forall (x,y) \in E$, $T(x,y)$ is less or
equal to the capacity of $(x,y)$. The {\em value} of $T$ is $\sum_{c}
T(i,c)=\sum_{s}T(s,f)$. 

\medskip

Let $\mbox{{\sc NCond}}_{\leq}$ be the set of inequalities obtained
from {\sc NCond} by replacing the strict inequalities by large
inequalities. 

\begin{lemma}\label{le-flow}
There exists a flow of value 1 in $\cN$ iff $(\mu_C,\mu_S)$ satisfies
$\mbox{{\sc NCond}}_{\leq}$. 
There exists a flow $T$ of value 1 such that $T(c,s)>0$ for all
$(c,s)\in E$ iff $(\mu_C,\mu_S)$ satisfies {\sc NCond}. 
\end{lemma}

The first part of Lemma \ref{le-flow} is proved in
\cite[Prop. 3.7]{CKW09}. We repeat the argument for
completeness. 

\begin{proof}
The 
celebrated Max-flow Min-cut Theorem~\cite{FoFu} states that the maximal
value of a 
flow is equal to the minimal capacity of a cut. Observe that the set of arcs $\{(i,c),c\in C\}$ forms a cut of
capacity 1. Therefore the maximal
flow is $\leq 1$ and it is 1 iff all cuts have a capacity $\geq 1$.

To be of finite capacity, a cut must contain only customer arcs $\{(i,c),c\in
C\}$ and server arcs $\{(s,f), s\in S\}$. Consider a subset $\cA = \{(i,c),c\in
C_1\}\cup \{(s,f), s\in S_1\}$. Set $C_2=C-C_1$ and $S_2=S-S_1$. The set $\cA$ is a cut iff $C_2\times
S_2 \cap E = \emptyset$, equivalently iff $S(C_2)\subset S_1$ and
$C(S_2)\subset C_1$. Also the capacity of $\cA$ is $\mu_C(C_1) +
\mu_S(S_1)$.

 \begin{figure}[htb]
   \centering
   \includegraphics[width=.3\textwidth]{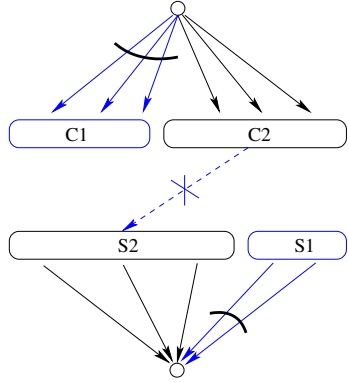}
   \caption{\label{fi:flow}Illustration of the proof of Lemma \ref{le-flow}.}
 \end{figure}

Assume that the cut $\{(i,c),c\in
C_1\}\cup \{(s,f), s\in S_1\}$ is of capacity strictly
less than 1. We have
\begin{equation*}
\mu_C(C_1) + \mu_S(S_1) < 1  
\iff  \mu_C(C_1) < \mu_S ( S_2)\:.
\end{equation*}
But $C(S_2)\subset C_1$ so, if $\mbox{{\sc NCond}}_{\leq}$ is satisfied,
we must have: 
\[
\mu_S ( S_2) \leq \mu_C ( C(S_2)) \leq \mu_C (C_1) \:.
\]
So we have proved that $\mbox{{\sc NCond}}_{\leq}$ is not satisfied. 

The other way round, if $\mbox{{\sc NCond}}_{\leq}$ is not satisfied,
then there exist $C_1,S_2,$ $C(S_2)=C_1$ such that $\mu_C(C_1) <
\mu_S(S_2)$. Set $C_2=C-C_1$ and $S_1=S-S_2$. By definition,
$C_2\times S_2 \cap E = \emptyset$, therefore $\{(i,c),c\in
C_1\}\cup \{(s,f), s\in S_1\}$ is a cut. Its capacity is $\mu_C(C_1) +
\mu_S(S_1) <1$. 

By contrapposing the above, we get that:
\[
\Bigl[ \mbox{{\sc NCond}}_{\leq} \mbox{ satisfied } \Bigr] \iff \Bigl[ \mbox{
  all cuts have a capacity } \geq 1  \Bigr] \iff \Bigl[ \mbox{ maximal
    flow is } 1 \Bigr]\:.
\]

We now prove the second part of the lemma. 
Assume that the conditions {\sc NCond} are not satisfied. If the 
conditions $\mbox{{\sc NCond}}_{\leq}$ are not satisfied either, then by the
first part of the proof there exists no flow of value 1. Assume now
that the conditions  $\mbox{{\sc NCond}}_{\leq}$ are satisfied. 
Then there
exists $U\subset C, U \neq C$, such that 
$\mu_C(U) = \mu_S(S(U))$.
Let
$T$ be any flow of value 1. Using the flow relation for $U$, we get:
\[
\sum_{(c,s)\in U \times S(U)} T(c,s) \ = \sum_{(i,c)\in \{i\} \times U}
T(i,c) = \mu_C(U)\:.
\]
Using $\mu_C(U) = \mu_S(S(U))$ and the flow relation for $S(U)$, we deduce that:
\[
\Bigl[ \sum_{(c,s)\in U \times S(U)} T(c,s) = \mu_S(S(U)) \ \Bigr] \quad \implies
\quad  \Bigl[  \sum_{(c,s)\in (C-U) \times S(U)} T(c,s) = 0 \ \Bigr] \:.
\]
Now it follows from the connectedness of $(C,S,E)$ that $(C-U) \times
S(U)\cap E \neq \emptyset$. We conclude that the flow $T$ is such
that $T(c,s) =0$ for some $(c,s)\in E$. 

\medskip

Assume now that the conditions {\sc NCond} are satisfied. 
Fix $\eta$ such that $0<\eta < 1/|E|$. Consider the function $T_{\eta}:A \rightarrow \R_+$ defined by 
\[
T_{\eta}(x,y)=\begin{cases}
\eta & \mbox{for } (x,y)=(c,s)\in E \\
|S(c)|\ \eta  & \mbox{for } (x,y)=(i,c) \\
|C(s)|\ \eta  & \mbox{for } (x,y)=(s,f) \:.
\end{cases}
\]
By construction $T_{\eta}$ is a flow. Set 
\begin{equation}\label{eq-tilde}
\widetilde{\mu}_C (c) = \frac{ \mu_C(c) - |S(c)|\eta }{1- |E|\eta}, \qquad
\widetilde{\mu}_S (s) = \frac{ \mu_S(s) - |C(s)|\eta }{1- |E|\eta} \:.
\end{equation}
For $\eta$ small enough, observe that $\widetilde{\mu}_C$, resp.  $\widetilde{\mu}_S$, is a
probability measure on $C$, resp. on $S$. 
Choose $\eta$ small enough such that $(\widetilde{\mu}_C,
\widetilde{\mu}_S)$ satisfies {\sc NCond}. This is possible since the
conditions {\sc NCond} are open conditions. 

Consider the directed graph $\cN$, see \eref{eq-dg}, with new capacities on the
customer and server arcs defined by $\widetilde{\mu}_C$ and 
$\widetilde{\mu}_S$. 
By applying the first part of the proof, there exists a flow
$\widetilde{T}: A \rightarrow \R_+$ of value 1. Define 
\[
T: A \rightarrow \R_+, \qquad T =  T_{\eta} + (1- |E|\eta)
\widetilde{T} \:.
\]
By construction $T$ is a flow for the graph $\cN$
with the original capacity constraints ($\mu_C$ for the customer arcs
and $\mu_S$ for the server arcs). The value of $T$ is 1 and it
satisfies $T(x,y)>0$ for all $(x,y)\in E$. 
This completes the proof.
\end{proof}

There exist algorithms to find the maximal flow which are polynomial
in the size of the underlying graph, independent of the  arc capacities. 
For instance, the classical ``augmenting path algorithm'' of Edmonds
$\&$ Karp~\cite{EdKa} operates in $O((|C|+|S|)|E|^2)$ time, and there
exist more sophisticated algorithms operating in $O((|C|+|S|)^3)$
time. 

Take one of these polynomial algorithms, call it {\sc
  MaxFlow} and consider it as a
blackbox. We build on this to design a polynomial algorithm to check
{\sc NCond}.  Let us detail the
construction.

\begin{lemma}\label{le-trick}
Define $(\widetilde{\mu}_C,\widetilde{\mu}_S)$ as in
\eref{eq-tilde}. The pair  $(\mu_C,\mu_S)$ satisfies  {\sc
  NCond} iff the pair $(\widetilde{\mu}_C,\widetilde{\mu}_S)$ satisfies  {\sc
  NCond} for  $\eta$ strictly positive and small enough. 
\end{lemma}

\begin{proof}
Assume that $(\mu_C,\mu_S)$ satisfies  {\sc
  NCond}. Since we are dealing with open conditions, any small enough
perturbation of  $(\mu_C,\mu_S)$ still satisfies {\sc
  NCond}. 

Assume now that $(\mu_C,\mu_S)$ does not satisfy  {\sc
  NCond}. There exists $U\subset C, U \neq C,$ such that $\mu_C(U)
\geq \mu_S(S(U))$. By using \eref{eq-tilde}, we get
\begin{eqnarray*}
\bigl(1-|E|\eta \bigr) \ \widetilde{\mu}_C(U) + \Bigl(\sum_{c\in U}|S(c)| \Bigr) \ \eta & \geq & 
\bigl(1-|E|\eta \bigr) \ \widetilde{\mu}_S(S(U)) + \Bigl(\sum_{s\in
  S(U)}|C(s)| \Bigr) \ \eta  \\
\bigl(1-|E|\eta \bigr) \ \widetilde{\mu}_C(U) + \bigl| E\cap (U\times
S(U)) \bigr| \ \eta & \geq &  
\bigl(1-|E|\eta \bigr) \ \widetilde{\mu}_S(S(U)) + \bigl| E\cap (C(S(U))\times S(U))\bigr| \  \eta \:.
\end{eqnarray*}
By definition, we have $U \subset C(S(U))$. We conclude that
$\widetilde{\mu}_C(U) \geq \widetilde{\mu}_S(S(U))$. So the pair
$(\widetilde{\mu}_C,\widetilde{\mu}_S)$ does not satisfy {\sc
  NCond}.
\end{proof}

Using Lemmas
\ref{le-flow} and \ref{le-trick}, {\sc NCond} is satisfied iff {\sc
  MaxFlow}$(\cN,\widetilde{\mu}_C,\widetilde{\mu}_S)$ returns 1
for $\eta$ small enough. 
So the trick is to run {\sc MaxFlow} on the input
$(\cN,\widetilde{\mu}_C,\widetilde{\mu}_S)$ by considering $\eta$ as a
formal parameter made ``as small as needed''. 

The precise meaning is the following. If 
$x_1,x_2,y_1,y_2 \in \R$, then: $(x_1+ y_1\eta) + (x_2+ y_2\eta) =
(x_1+x_2) + (y_1+y_2)\eta$. Furthermore,
\begin{eqnarray}\label{eq-eta}
\bigl[ x_1 + y_1\eta = x_2 + y_2 \eta \bigr] & \iff &  \bigl[
  x_1=x_2, \ y_1 = y_2 \bigr]  \nonumber \\
\bigl[ x_1 + y_1\eta < x_2 + y_2 \eta \bigr] & \iff &  \Bigl[
  (x_1< x_2) \ \mbox{ or } \  (x_1 = x_2, \ y_1 < y_2) \Bigr]
\end{eqnarray}
So $\eta$ is small enough not to reverse any strict inequality. When
running {\sc  MaxFlow} on
$(\cN,\widetilde{\mu}_C,\widetilde{\mu}_S)$, the algorithm deals with
values of the type $(x+y\eta)$, and adds and compare them according to
the above rules. Now observe that the algorithm stops in finite time,
so it will have performed only a finite number of
operations. Therefore, it would be possible, a posteriori, to assign to
$\eta$ a value which would be small enough to enforce \eref{eq-eta}. 

\medskip

\begin{algorithm}
\KwData{$(C, S, E)$, $(\mu_C,\mu_S)$ such that $\mbox{supp}(\mu_C)=C$,
  $\mbox{supp}(\mu_S)=S$. 
}
\KwResult{``Yes'' if {\sc NCond}, ``No'' if $\neg$({\sc NCond})}
\Begin{Compute $\cN$, $\widetilde{\mu}_C$, $\widetilde{\mu}_S$\;
\eIf{{\sc MaxFlow}$(\cN,\widetilde{\mu}_C,\widetilde{\mu}_S) =1$}{{\sc
    Result} $\leftarrow$ Yes\;}
{{\sc
    Result} $\leftarrow$ No\;}
}
\Return{{\sc Result}\;}
\caption{Checking the necessary stability conditions}
\label{al:ncond}
\end{algorithm}

The termination is obvious and the correctness follows from Lemmas
\ref{le-flow} and \ref{le-trick}. 

\begin{proposition}
Given a bipartite model $[(C,S,E),\mu]$, there exists an algorithm of time complexity $O((|C|+|S|)^3)$ to
decide if {\sc NCond} is satisfied. 
\end{proposition}

\section{Connectivity properties of the Markov chain}\label{se-utc}

Define the following
property for the transition graph of the Markov chain:
\begin{center}
{\sc UTC} : a unique (terminal) strictly connected component with all states leading
to it. 
\end{center}
Property {\sc UTC} is necessary for stability as
defined in Def. \ref{de-stab}. 
But property {\sc UTC} is not granted in bipartite matching models and counterexamples are given
below (Examples \ref{ex-nonirred1} and \ref{ex-nonirred2}). In fact, we will see that we are in an unusual situation: the
necessary stability conditions {\sc NCond} turn out to be sufficient
conditions for the property {\sc UTC} (Theorem \ref{th-P})! 
Observe also that property {\sc UTC} is weaker than irreducibility, and
we will give an example of a model satisfying {\sc NCond} and {\sc UTC} without being
irreducible (Example \ref{ex-Pnonirred}). 

\subsection{Stable structures}

To establish property {\sc UTC}, we make a detour by introducing and
studying a notion of independent interest: stable structures. 

\begin{definition}
A bipartite matching structure $(C,S,E,F)$ is {\em stable}
if there exists a probability measure $\mu$ satisfying
(\ref{eq-3cond}) and  whose marginals
$\mu_C$ and $\mu_S$ satisfy {\sc NCond}. 
\end{definition}

The justification for this terminology will appear in \S \ref{se-ml}:
we prove there that under the ML policy, any model
satisfying {\sc NCond} is stable. So a  structure is stable iff there 
exists an associated model which is stable.

\medskip

First of all, there exist stable structures. 

\begin{ex}\label{ex-irred}
Consider $(C,S,E,C\times S)$, where $(C,S,E)$ is the NN bipartite
graph of Figure \ref{fig:NN}. Let
\[
\mu_C: \ \mu_C(1)=\mu_C(2)=2/5, \ \mu_C(3)=1/5, \qquad \mu_S:
\ \mu_S(1')=\mu_S(2')=2/5, \  \mu_S(3')=1/5 \:. 
\]
The product measure $\mu=\mu_C \times\mu_S$ has marginals $\mu_C$ and $\mu_S$ 
and we check that $(\mu_C,\mu_S)$ satisfy {\sc NCond}. Also it is
easily proved that for any admissible matching policy, the graph of
the Markov chain is irreducible. 
\end{ex}

On the other hand, there exist unstable structures. 
We illustrate this on two examples. 

\begin{ex}\label{ex-nonirred1}
Consider the structure $(C,S,E,F)$ where $(C,S,E)$ is the NN graph of
Figure \ref{fig:NN}, and where 
\[
F = \bigl\{ (1,3'), (2,2'), (3,1') \bigr\}\:.
\]
Consider any $\mu$ with $\mbox{supp}(\mu)=F$. We have
$\mu_C(1)=\mu_S(3')=\mu(1,3')$ which violates {\sc NCond} for
$V=\{3'\}$. We can also prove that the property {\sc UTC} is not
satisfied. 
Consider a state of the type $(x,y)$ with $x=y=(0,0,k)$, for some $k\geq 0$. Any one of the three possible arrivals leave 
the state unchanged. In particular, there is an infinite number of
terminal components. 
\end{ex}

\begin{ex}\label{ex-nonirred2}
Consider the bipartite matching structure defined in Figure
\ref{fig:non-irred}. 
The graph $(C,S,E)$ is represented on the left of
the figure, while the graph  $(C,S,F)$ is represented on the right.

 \begin{figure}[htb]
   \centering
   \includegraphics[width=.6\textwidth]{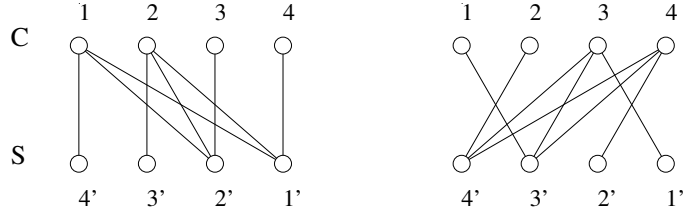}
   \caption{\label{fig:non-irred}The matching graph $(C,S,E)$ on the
     left, and the arrival graph $(C,S,F)$ on the right.}
 \end{figure}

Consider any $\mu$ with $\mbox{supp}(\mu)=F$. We have 
\[
\mu_S(\{1',2'\}) = 
\mu(3,1') + \mu(4,2') \leq \mu_C(\{3,4\})\:,
\]
which contradicts  {\sc NCond} for $U=\{3,4\}$. We can also prove that
the property {\sc UTC} is not satisfied. 
Consider  a state $(x,y)$ with $x_3+x_4 = k >0$. 
Reducing the number of customers of types 3/4 would require an arrival of type $(1,1')$ or $(1,2')$ or
$(2,1')$ or $(2,2')$. But none of these pairs belong to $F$. Therefore
it is impossible to reach a state $(x',y')$ with $x'_3+x'_4
<x_3+x_4$. On the other hand an arrival of type $(3,3')$ or $(3,4')$ or
$(4,3')$ or $(4,4')$ strictly increases the number of customers of
types 3/4. Hence all the states are transient, and there is no
terminal strongly connected component. 
\end{ex}

Stability of a structure is a decidable property. There 
exists a probability measure $\mu$ with the requested properties iff
the following system of linear inequalities in the
indeterminates $\mu(c,s), \ c\in C, s\in S,$ have a solution:
\begin{equation}
\begin{cases}
 \sum_{(c,s)\in C\times S} \mu(c,s) = 1, & \\
 \mu(c,s) > 0, & \forall (c,s) \in F,\\
 \mu(c,s) =0, & \forall (c,s) \in C\times S - F,\\
\mu_C(c) = \sum_{s\in S} \mu(c,s),  & \forall c\in C, \\
\mu_S(s) = \sum_{c\in C} \mu(c,s), &  \forall s \in S, \\
\mbox{{\sc NCond}} \:.
\end{cases}
\end{equation}
However, the number of inequalities  is exponential in $|C| +
|S|$. 
We are going to propose a  criterion which is much simpler, both
conceptually and algorithmically. 

\medskip

Consider a bipartite matching structure $(C,S,E,F)$. Define
$\widetilde{F}= \{(s,c) \mid (c,s)\in F\}$. Associate with the
structure the {\bf directed} graph $(C\cup S, E\cup \widetilde{F})$, in
other words the nodes are $C\cup S$ and the arcs are
\begin{equation*}
c \longrightarrow s, \quad  \mbox{if }  (c,s) \in E, \qquad 
s \longrightarrow c, \quad  \mbox{if }  (c,s) \in F \:.
\end{equation*}

We have represented in Figure \ref{fig:non-irred2} the directed graph
associated with the structure of Example~\ref{ex-nonirred2}. 

 \begin{figure}[htb]
   \centering
   \includegraphics[width=.4\textwidth]{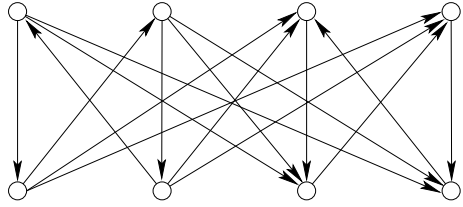}
   \caption{\label{fig:non-irred2}The directed graph associated
     with the structure of Figure \ref{fig:non-irred}.}
 \end{figure}

The graph of Figure \ref{fig:non-irred2} is not strongly connected: the four nodes on
the right  form a strongly connected
component. Similarly, the directed graph associated with the structure
of Example \ref{ex-nonirred1} is not strongly connected. On the other
hand, the directed graph associated with the structure of
Example \ref{ex-irred} is strongly connected. 
This is not a coincidence. 

\begin{theorem}\label{th-2}
Let $(C,S,E,F)$ be a bipartite matching structure. The following two properties are equivalent: 
\begin{enumerate}
\item $(C,S,E,F)$ is a stable structure; 
\item $(C\cup S, E\cup\widetilde{F})$ is strongly connected. 
\end{enumerate}
\end{theorem}

In particular, one can decide if a structure is  stable with an
algorithm of time complexity $O(|C||S|)$ by testing the strong
connectivity of $(C\cup S, E\cup\widetilde{F})$. 

\medskip

\begin{proof}[Proof of Theorem \ref{th-2}]
Assume that $(C,S,E,F)$ is a stable structure. Let $\mu$ be a
probability measure satisfying (\ref{eq-3cond}) and {\sc NCond}. 
Suppose that there exist $c\in C, s\in
S$, with no directed
path from $c$ to $s$ in $(C\cup S, E\cup\widetilde{F})$. Let
$\mbox{succ}(c)$ be the set of nodes that can be reached starting from
$c$ in $(C\cup S, E\cup\widetilde{F})$. Set
\[
C_1= C \cap \mbox{succ}(c), \ S_1 = S \cap
\mbox{succ}(c), \ C_2 = C - C_1, \ S_2 = S - S_1 \:.
\]
By assumption, $s\in S_2$. The following two properties hold:
\begin{equation*}
\mu(C_2,S_1) = 0, \qquad (C_1\times S_2)\cap E = \emptyset \:.
\end{equation*} 
Using $\mu(C_2,S_1) = 0$, we get
\[
\mu_S(S_1)= \mu (C_1,S_1) \leq \mu_C(C_1) \:.
\]
But using $[ (C_1\times S_2)\cap E = \emptyset]$ and {\sc NCond} for
$U=C_1$, we get 
\[
\mu_C(C_1) < \mu_S(S(C_1)) = \mu_S(S_1) \:.
\]
{From} this contradiction, we deduce that for all $c\in C, s\in S$, there exists a directed
path from $c$ to $s$ in $(C\cup S, E\cup\widetilde{F})$. Similarly, we can prove that for all $s\in S, c\in
C$, there exists a directed
path from $s$ to $c$ in $(C\cup S, E\cup\widetilde{F})$. 

\medskip

Assume now that $(C\cup S, E\cup\widetilde{F})$ is strongly
connected. 
Consider the matrices $A\in \R_+^{C\times S}$ and $B\in \R_+^{S\times
  C}$ defined by 
\[
A_{cs} = \begin{cases} 1/|S(c)| & \mbox{if } (c,s) \in E \\
                         0 & \mbox{otherwise}
\end{cases}, \qquad B_{sc} = \begin{cases} 1/\#\{d, (s,d)\in \widetilde{F}\} & \mbox{if } (s,c) \in \widetilde{F} \\
                         0 & \mbox{otherwise}
\end{cases}
\:.
\]
Consider the matrix $AB \in \R_+^{C\times C}$. By construction, we
have $(AB)_{cd}>0$ if and only if there is a path of length 2 from $c$
to $d$ in the graph $(C\cup S, E\cup\widetilde{F})$. Since $(C\cup S,
E\cup\widetilde{F})$ is strongly connected, we deduce that $AB$ is
irreducible. Clearly the spectral radius of $AB$ is 1. Applying the Perron-Frobenius Theorem~\cite{sene}, we
obtain the existence of a line vector $x\in \R_+^C$ such that: $\forall
c, x_c>0$, $\sum_c x_c =1$, and $xAB=x$.  Set $y=xA$. Define the
probability measure $\mu$ on $C\times S$ by $\mu(c,s)= y_sB_{sc}$. 
By construction, we have $\mu_C=x, \ \mu_S=y$. Also, by construction,
$\mbox{supp}(\mu)=F$. 
Define the function
$T:A \rightarrow \R_+$ by 
\[
\forall c\in C, \ T(i,c) = x_c, \quad \forall s \in S, \ T(s,f) =
y_s, \quad \forall (c,s) \in E, \ T(c,s) = x_cA_{cs}\:.
\]
By construction, $T$ is a flow of value 1 such that $T(c,s)>0$ for all
$(c,s)\in E$. Using Lemma \ref{le-flow}, we get that
$(x,y)=(\mu_C,\mu_S)$ satisfies {\sc NCond}. 
\end{proof}

\subsection{Back to property {\sc UTC}}

We now have all the ingredients needed to prove the following result. 

\begin{theorem}\label{th-P}
Consider a bipartite matching model $[(C,S,E),\mu,\mbox{{\sc Pol}}]$. 
Assume that the structure $(C,S,E,F)$ is stable, equivalently
that $(C\cup S, E\cup\widetilde{F})$ is strongly connected. Then 
the transition
graph of the Markov chain of the bipartite matching model satisfies the
property {\sc UTC}.
\end{theorem}

\begin{proof}[Proof of Theorem \ref{th-P}]
We are going to prove that the empty state can be reached
starting from any state. This is a sufficient condition for property
{\sc UTC} to hold. The unique terminal strongly connected component is
the set of states that can be reached from the empty state. 

\medskip

We carry out the proof in the
commutative case, but it works unchanged in the non-commutative
case (the only information needed is the number of customer/server of
each class). Consider a non-empty state $(X,Y)$, with $X=(x_c)_{c\in C}$ and $Y=
(y_s)_{s\in S}$. It is sufficient to prove that we can always reach a
state $(X',Y')$ such that $|X'|<|X|$. 

If there exists $(c,s)\in \Cz\times \Sz$ such that
$\mu(c,s) > 0$, then the proof is completed. 
Assume now that $\mu(\Cz\times \Sz) = 0$. Choose $(c,s)\in
\Cz\times \Sz$. By assumption, there exists a path from
$s$ to $c$ in $(C\cup S, E\cup\widetilde{F})$. Let us denote it by
$(s=s_1,c_1,s_2,c_2,\dots, s_k,c_k=c)$. 
Assume that $c_1,\dots, c_{k-1}\not\in \Cz$ and 
$s_2,\dots, s_{k}\not\in \Sz$. 
(If not, consider a subpath with this property.) Assume further that $c_1 \in \Cr$.
(If $c_1 \in \Cb$, then $(c_1, s_2) \in E$ implies $s_2 \in \Cz$, so we can consider the subpath
$(s=s_2,c_2,\dots, s_k,c_k=c)$.) 
Since $(s_i,c_i) \in \widetilde{F}$ and $(c_i,s_{i+1}) \in E$ by construction and since $c_1 \in \Cr$, we get that 
$c_1,\dots, c_{k-1}\in \Cr$ and $s_2,\dots, s_{k}\in \Sr$.

 \begin{figure}[htb]
   \centering
   \includegraphics[width=.9\textwidth]{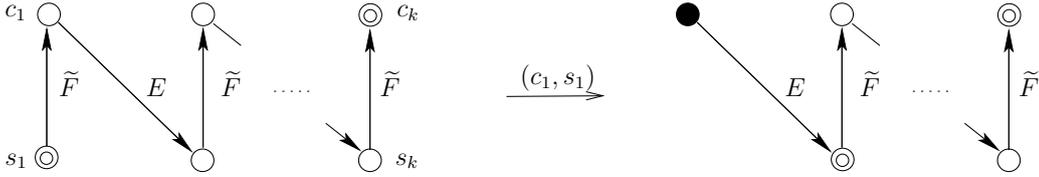}
   \caption{\label{fi:path}The path $(s=s_1,c_1,s_2,c_2,\dots,
     s_k,c_k=c)$ in $(C\cup S, E\cup\widetilde{F})$.}
 \end{figure}

By definition of the graph
$(C\cup S, E\cup\widetilde{F})$, we have $\mu(c_i,s_i)>0$ for all $i$. 
Choose the sequence of arrivals $(c_1,s_1), \dots, (c_k,s_k)$. 
Consider the effect of the arrival of $(c_1,s_1)$. Since $s_1\in
\Sz$, it will be matched with a customer of $C_{\bullet}$ (and
not with $c_1$, even
if $(c_1,s_1)\in E$, since an admissible matching policy is always {\em
  buffer-first}, see \S \ref{sse-match}). Since $c_1\not\in \Cz$, it will remain
unmatched. Let $(X^{(1)},Y^{(1)})$ be the new state. We have
$|X^{(1)}|=|X|$. Also, in the new state, we have $c_1\in C^{(1)}_{\bullet}$,
which implies that $s_2 \in \Sz^{(1)}$. So we can repeat the
argument inductively. After the arrivals of $(c_1,s_1), \dots,
(c_{k-1},s_{k-1})$, we are in a state $(X^{(k-1)},Y^{(k-1)})$ satisfying:
\[
|X^{(k-1)}|=|X|, \quad s_k \in \Sz^{(k-1)}, \quad c_k \in
\Cz^{(k-1)}\:.
\]
Therefore, after the arrival of $(c_k,s_k)$, we end up in a state
$(X^{(k)},Y^{(k)})$ such that $|X^{(k)}| = |X| -1$. 
This completes the proof.
\end{proof}

\begin{ex}\label{ex-Pnonirred}
Consider a bipartite matching model associated with the structure
$(C,S,E,F)$ where $(C,S,E)$ is the NN graph of 
Figure \ref{fig:NN}, and where 
\[
F = \bigl\{ (1,1'), (2,2'), (3,3') \bigr\}\:.
\]
The graph $(C\cup S, E\cup\widetilde{F})$ is strongly
connected. According to Theorem \ref{th-2}, the graph satisfies
property {\sc UTC}. But it is not irreducible. Indeed, it is impossible to
reach the state $((0,1,0);(0,0,1))$ starting from the empty
state. More generally, none of the states of the facet
$(\{2\},\{3'\})$ belong to the terminal strongly connected component. 
\end{ex}

Below, we study the stability of bipartite matching models. 
Therefore, we always assume that the necessary conditions {\sc NCond}
are satisfied. So we get the property {\sc UTC} for the Markov chain as a
consequence of Theorem \ref{th-P}.

\section{Models that are stable for all admissible policies}
 \label{se-suff}

\begin{definition}
Consider a bipartite graph $(C,S,E)$ and an admissible matching
policy {\sc Pol}. The {\em  stability region} is the set of values of $\mu$ 
for which 
the bipartite matching model $[(C,S,E),\mu,\mbox{{\sc Pol}}]$ is stable. 
\end{definition}

The stability region is included in the polyhedron defined by {\sc NCond}. 
The stability region is {\em maximal} if it is equal to this polyhedron. 

\medskip

Denote by ${\mathfrak F}$ the set of 
facets. 
Define the following conditions on $\mu$:

\begin{equation}
\label{eq:SNat}
\textrm{{\sc SCond}}: \quad 
\begin{array}{ll}
\mu_C(\Cz(\cF)) + \mu_S(\Sz(\cF)) \ > \ 1 - \mu(E \cap \Cr(\cF)
\times \Sr(\cF)), \quad \forall \cF \in {\mathfrak F}-\{\emptyset\}
\end{array}
\end{equation}

\medskip

Let $\cF$ be a saturated facet, see Definition \ref{de-saturated}. 
Assume for instance that $\Cr(\cF)=\emptyset$. 
Then $E\cap \Cr(\cF) \times \Sr(\cF) = \emptyset$ and $\Cz(\cF)=C -
\Cb(\cF)$. So \eref{eq:SNat} implies: 
\[
\mu_S(\Sz(\cF)) \ > \ \mu_C(\Cb(\cF)) \:.
\]
Since $\Sz(\cF)=S(\Cb(\cF))$, we recognize exactly \eref{eq:CNat} for
$U=\Cb(\cF)$. Conversely, consider $U\subsetneq C$ and the associated
condition in {\sc NCond}: $\mu_C(U) < \mu_S(S(U))$. 
Choose a state with a strictly positive number of customers/servers
for the classes $U$ and  $S-S(U)$.
Let $\cF$ be the corresponding facet. The facet $\cF$ is
saturated: $\Sz (\cF)=S(U), \ \Sb (\cF) = S -S(U), \ \Sr (\cF)
=\emptyset$. Let us apply \eref{eq:SNat} to the facet $\cF$, we get: 
\begin{eqnarray*}
\mu_C(\Cz(\cF)) + \mu_S(S(U)) & > & 1  \\
           \mu_S(S(U))      & > & \mu_C(C - \Cz(\cF)) \ \ \geq \ \ \mu_C(U) \:.
\end{eqnarray*}

To summarize, the
subset of the inequalities \eref{eq:SNat}
obtained by considering only the saturated facets gives precisely the
inequalities {\sc NCond}. 

\medskip

We now show that the conditions {\sc SCond} are sufficient stability
conditions. 

\begin{proposition}
\label{pr:CF}
A bipartite model with probability $\mu$ satisfying {\sc SCond} is
stable under any admissible matching policy.   
\end{proposition}

\proof
Consider 
the linear Lyapunov function: 
$$
L(u,v) = |u|, \; (u,v) \in \CE,
$$
the number of unmatched customers (servers). 
Let $(U_n,V_n)_n$ be the Markov chain of the buffer-content. 
Let $\cF \not= \emptyset$ be an arbitrary and fixed facet. 
Then for any $(u,v) \in \cF$ we have (see Table~\ref{tab:LyapLin}):
\begin{eqnarray*}
&& \expct{L(U_{n+1},V_{n+1}) \; | \; (U_n, V_n) = (u,v)} - L(u,v) \ =  \ 
- \mu(\Cz(\cF), \Sz(\cF)) +  \mu(\Cr(\cF), \Sb(\cF)) \\
&&   +~\mu(\Cb(\cF), \Sb(\cF)) + \mu(\Cb(\cF), \Sr(\cF)) +  \mu(\Cr(\cF) \times \Sr(\cF) \cap E^c)\\
& = & 1 - \mu_C(\Cz(\cF)) - \mu_S(\Sz(\cF)) - \mu( \Cr(\cF) \times
\Sr(\cF) \cap E) \:.
\end{eqnarray*}
The inequality (\ref{eq:SNat}) implies directly
that: 
\begin{equation}\label{eq-foster}
\expct{L(U_{n+1},V_{n+1}) \; | \; (U_n, V_n) = (u,v)} - L(u,v) < \epsilon < 0 \:.
\end{equation}

By application of the Lyapunov-Foster Theorem, see for instance
\cite[\S 5.1]{brem99}, we conclude that the model is stable. 
\begin{table}[hbt]
\centering
\begin{tabular}{c|ccc}
               & $\Cz$ & $\Cr$ & $\Cb$ \\
\hline 
$\Sz$ & $-1$ & $0$ & $0$ \\
$\Sr$  & $0$ & $0$ or $1$ & $1$ \\
$\Sb$ & $0$ & $1$ & $1$
\end{tabular}
\caption{\label{tab:LyapLin}Variation of the linear Lyapunov function.}
\end{table}
\endproof

\begin{corollary}
Consider a bipartite graph in which any non-zero facet is
saturated. For any admissible matching policy, the stability region is 
maximal. 
\end{corollary}

The bipartite graph $(C=\{1,2\},S=\{1',2'\}, C\times S - \{(2,2')\})$
is such that any non-zero facet is
saturated. Therefore, its stability region is maximal for any admissible policy.  
The same is true for the ``almost complete graphs''
$(C=\{1,\dots, k\}, S=\{1',\dots, k'\}, C\times S -\{(i,i'),\forall
i\})$.

\begin{ex}
\label{ex:NNscond}
Consider the NN graph from Figure \ref{fig:NN}. The graph has only one non-zero facet
that is non-saturated, facet $(\{3\},\{3'\})$. 
For any admissible policy, the stability region is at least the polyhedron {\sc SCond},
Proposition \ref{pr:CF}, which is defined by:
\begin{equation}\label{eq-region}
\mbox{{\sc NCond}}, \qquad \mu_C(1) + \mu_S(1') > 1 - \mu(2,2') \:.
\end{equation}
Assume now $\mu = \mu_C\times\mu_S$ and $\mu_C = \mu_S$. Set $x=\mu_C(1) = \mu_S(1')$ and 
$y=\mu_C(2) = \mu_S(2')$. Then:
$$
\textrm{{\sc NCond}}: \quad \left \{ 
\begin{array}{ll}
x < 0.5 \\
2x + y  > 1
\end{array}
\right .
\quad \quad
\textrm{{\sc SCond}}: \quad \left \{ 
\begin{array}{ll}
\textrm{{\sc NCond}} \\
2x + y^2 > 1
\end{array}
\right .
$$
\begin{figure}[htb]
  \centering
  \includegraphics[width=.35\textwidth]{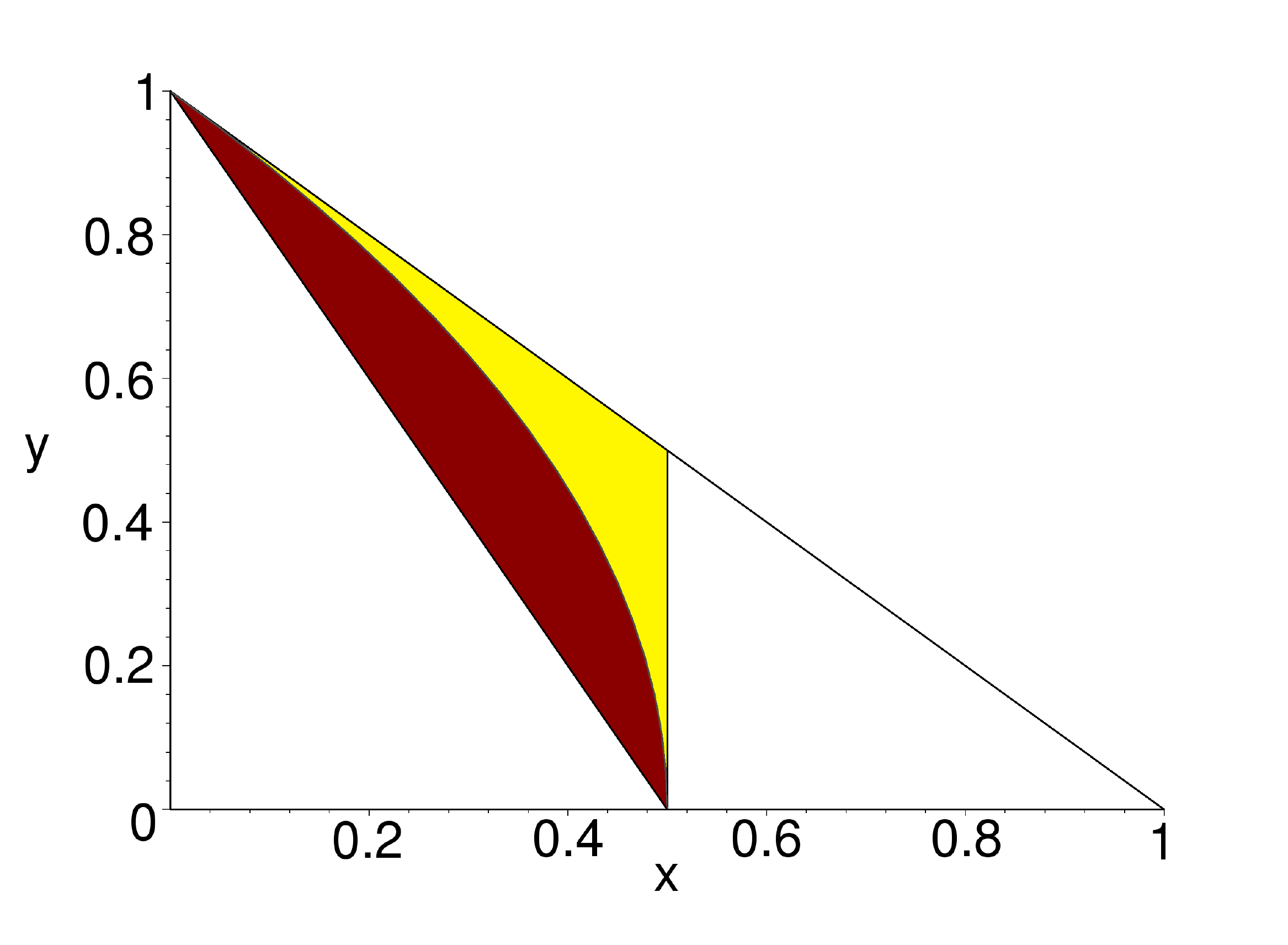}
  \caption{\label{fig:scond}{\sc NCond} and {\sc SCond} for the NN-graph with $\mu = \mu_C \times \mu_S$ and $\mu_C = \mu_S$.}
\end{figure}
In Figure \ref{fig:scond}, the light (yellow) region corresponds to {\sc SCond}, 
and the union of the light and dark~(red) regions corresponds to {\sc NCond}. 
\end{ex}

Unfortunately, for some bipartite graphs, the polyhedron 
{\sc SCond} is empty. This is illustrated by 
the following example.

\begin{ex}
Consider the NNN graph of Figure \ref{fig:NNN}. 
The condition {\sc SCond} for facet $(\{1\}, \{4'\})$ gives:
\begin{equation}
\label{eq:NNN14}
\mu_C(\{3, 4\}) + \mu_S(\{1', 2'\})  > 1 - \mu(2, 3')
\end{equation}
and for facet  $(\{4\}, \{1'\})$:
\begin{equation}
\label{eq:NNN41}
\mu_C(1) + \mu_S(4') > 1 - \mu(2, 2') - \mu(2, 3')- \mu(3, 3').
\end{equation}
The inequality (\ref{eq:NNN14}) is equivalent to: 
$
\mu_C(1)  + \mu_S(4')  <  1 - \mu_C(2) - \mu(\{1,3, 4\}, 3').
$
Together with (\ref{eq:NNN41}) this gives:
$$
\mu_C(1)  + \mu_S(4')  <  1 - \mu_C(2) - \mu(\{1,3, 4\}, 3') < 
1 - \mu(2, 2') - \mu(2, 3')- \mu(3, 3') < \mu_C(1)  + \mu_S(4'),
$$
which is impossible.

 \begin{figure}[htb]
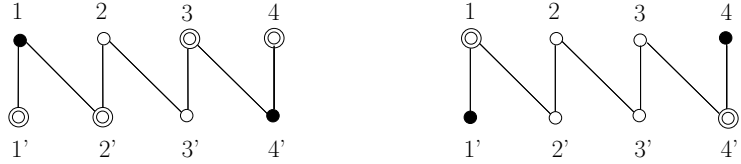

  \centering
  \includegraphics[width=.25\textwidth]{NNNf14.pdftex_t} \qquad \qquad \qquad
  \includegraphics[width=.25\textwidth]{NNNf41.pdftex_t}
  \caption{\label{fig:NNN}NNN graph: facets $(\{1\}, \{4'\})$ and $(\{4\}, \{1'\})$.}
\end{figure}
\end{ex}
\begin{figure}[htb]
  \centering
  \includegraphics[width=\textwidth]{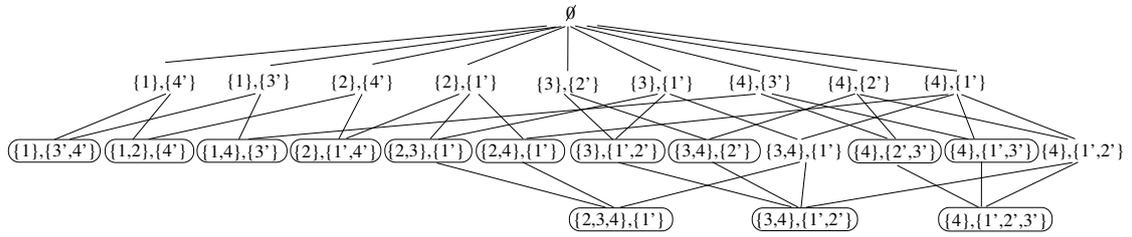}
  \caption{\label{fig:NNNfacets}Facets for the NNN graph. Saturated facets are encircled  
($13$ among $25$ facets).}
\end{figure}

\section{Priorities and MS are not always stable}\label{se-ms}

Consider the NN bipartite graph of Figure \ref{fig:NN} and Example \ref{ex:NNscond}. 
For this model, Proposition \ref{pr:CF} does not allow to decide if the stability region is maximal (see Figure \ref{fig:scond}).
In Figure \ref{fig:zoneMS}, we give simulation results for the average buffer size up to time $n=1000000$ for 
the NN-graph with $\mu = \mu_C \times \mu_S$, $\mu_C = \mu_S$, and MS policy. 
\begin{figure}[htb]
  \centering
  \includegraphics[width=.45\textwidth]{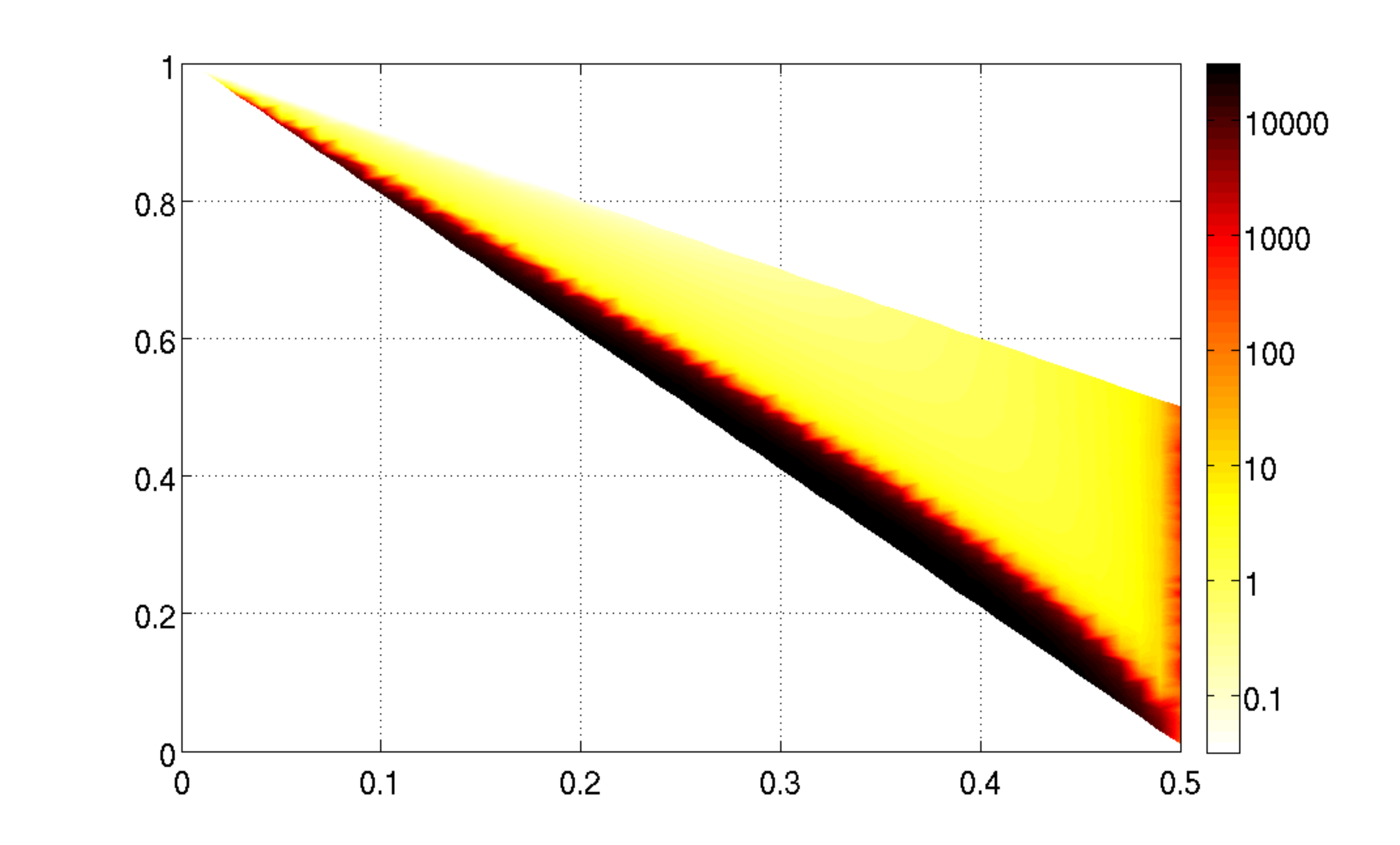}
  \caption{\label{fig:zoneMS}Average buffer size for the NN-graph 
with $\mu = \mu_C \times \mu_S$, $\mu_C = \mu_S$, and MS policy.}
\end{figure} 
We can see that the average buffer size is growing rapidly near the
$2x+y=1$ line. This does not necessarily imply unstability, as even for
stable models we could have the mean stationary buffer size that is
growing unboundedly as we approach the boundary of the stability
region.  
In fact, we show below that for the PR and MS matching policies, the
stability region is not maximal. 

\begin{proposition}
Consider the NN model with either the MS policy or the PR policy given by: 
 $$
A = \left [
\begin{array}{ccc}
 0 & 2 & 1 \\
 2 & 1 & 0 \\
 1 & 0 & 0 
\end{array}
\right ] \; \textrm{and} \;  
\quad B = \left [
\begin{array}{ccc}
 0 & 2 & 1 \\
 2 & 1 & 0 \\
 1 & 0 & 0 
\end{array} 
\right ].
 $$
For both policies, the stability region is not maximal.
\end{proposition}

\proof
We carry out the proof for the PR policy. The idea of the proof is to
play with two different Markov chains: the one describing the evolution
of the buffer content, and an auxiliary one which mimicks the
evolution of the customers/servers of type 2/2' on some of the
facets. 

\medskip

Consider an auxiliary Markov chain on $\Z$ with transition probabilities:

\begin{center}
\begin{tabular}{c|ccc}
               & $x \rightarrow x-1$ & $x \rightarrow x$ & $x \rightarrow x+1$ \\
\hline 
$x < 0$ & $a_{-1}$ & $a_0$ & $a_{1}$ \\
$x = 0$ & $b_{-1}$ & $b_0$ & $b_1$ \\
$x > 0$ & $c_{-1}$ & $c_0$ & $c_1$
\end{tabular}
\end{center}

Assume that $a_{-1},a_1,b_{-1},b_1,c_{-1},c_1$, are all different from
0. The chain is positive recurrent iff: $\bigl[ a_{-1} < a_{1}, \ c_1 <
  c_{-1} \bigr]$. The stationary distribution is then equal to:
\begin{eqnarray*} 
 \pi(0) & = & \Bigl( 1 + \frac{b_{-1}}{a_1 - a_{-1}} + \frac{b_1}{c_{-1} - c_1}  \Bigr)^{-1},\\
 \pi(x) & = & 
  \pi(0)\frac{b_1}{c_1}\Bigl(\frac{c_1}{c_{-1}}\Bigr)^x \ \ \qquad \mbox{if } x > 0,\\
  \pi(x) & = &  \pi(0)\frac{b_{-1}}{a_{-1}}
  \Bigl(\frac{a_{-1}}{a_1}\Bigr)^{|x|} \qquad   \mbox{if }  x < 0  \:.
\end{eqnarray*}

Consider a NN-model with a probability $\mu$ such that $\mbox{supp}(\mu) = C \times S$.
Let $(X,Y)=(X(n), Y(n))_{n}$ be the Markov chain of the buffer content, where
$X(n) = (X_1(n), X_2(n), X_3(n))$ and $Y(n) = (Y_1(n), Y_2(n),
Y_3(n))$.  Assume wlog that $(X,Y)$ is given under the form of a Stochastic
Recursive Sequence, that is:
\[
(X(n+1),Y(n+1)) = \Phi ( X(n), Y(n), \theta_n)\:,
\]
where $(\theta_n)_n$ is an i.i.d. sequence of r.v.'s distributed
according to $\mu$, and $\Phi$ is a deterministic function. 

Consider now the process $(X_2(n)-Y_2(n))_n$. This is not a Markov
chain. However, if $X_2(n)+X_3(n)>0$, then ``it becomes one''. More
precisely, if $X_2(n)+X_3(n)>0$, then 
\begin{equation}\label{eq-mc}
X_2(n+1)-Y_2(n+1) = \Phi_2 ( X_2(n)-Y_2(n), \theta_n) \:,
\end{equation}
where $\Phi_2$ is a deterministic function. This can be checked by
direct inspection. Moreover, the transition kernel on $\Z$ defined by the
recursion \eref{eq-mc} is of the type of the above 
auxiliary chain
with parameters: 
\begin{eqnarray*}
a_1 & = &\mu(1,1') + \mu(1,3') + \mu(2,1') + \mu(2,3'),\\
a_0 & = &\mu(1,2') + \mu(2,2') + \mu(3,1') + \mu(3,3'),\\
a_{-1} & = & \mu(3,2'),\\
b_1 & = &\mu(2,1') + \mu(2,3'),\\
b_0 & = &\mu(1,1') + \mu(1,3') + \mu(2,2') + \mu(3,1') + \mu(3,3'),\\
b_{-1} & = &\mu(1,2') + \mu(3,2'),\\
c_1 & = & \mu(2,3'),\\
c_0 & = &\mu(1,3') + \mu(2,1') + \mu(2,2') + \mu(3,3'),\\
c_{-1} & = &\mu(1,1') + \mu(1,2') + \mu(3,1') + \mu(3,2').
\end{eqnarray*}
Let us justify for instance the values of $c_{-1},c_0,c_1$. 
We are in the case $X_2(n)+X_3(n)>0$ and $X_2(n)-Y_2(n)>0$ which implies: 
\[
X_1(n)=0, X_2(n) >0, \quad  Y_1(n)=Y_2(n)=0,
Y_3(n)>0 \:,
\]
In Table \ref{tab:MS} below, we show the effect of the different possible types
of arrivals, restricting to the ones which may affect $X_2$, i.e. when the
customer class is 2 or the server class is $1'$ or $2'$. To simplify, we
have assumed in Table \ref{tab:MS} that $X_3>0$. For $X_3=0$, the ``Possible
  matchings'' column would be affected, but not the ``Selected
  matchings'' and ``$\Delta X_2$'' columns. 

\begin{table}[htb]
\begin{center}
\begin{tabular}{c|c|c|c}
   Arrival            & Possible matchings & Selected matchings &
   $\Delta X_2$ \\
\hline 
$(1,1')$ & $(1,3')$, $(2,1')$, $(3,1')$ & $(1,3')$, $(2,1')$ & -1 \\
$(1,2')$ & $(1,3')$, $(1,2')$, $(2,2')$ & $(1,3')$, $(2,2')$ & -1 \\
$(2,1')$ & $(2,1')$, $(3,1')$  & $(2,1')$ & $0$ \\
$(2,2')$  & $(2,2')$  & $(2,2')$ & $0$ \\
$(2,3')$ & $\emptyset$ & $\emptyset$ & $+1$ \\
$(3,2')$ & $(2,2')$ & $(2,2')$ & -1 \\
$(3,1')$ & $(2,1')$, $(3,1')$ & $(2,1')$ & -1 
\end{tabular}
\caption{\label{tab:MS}Effect of arrivals.}
\end{center}
\end{table}

Let us comment on a couple of cases. If the arrival is of type
$(1,1')$, then the selected matching is $(2,1')$ rather than $(3,1')$
due to the PR policy ($B_{2, 1'} > B_{3, 1'}$). 
If the arrival is of type $(1,2')$, the selected matching is $(2,2')$
rather than $(1,2')$ according to the buffer-first
property of admissible policies, see \S \ref{sse-match}. 
The other cases are argued similarly. 

\medskip

Let us introduce a new Markov chain $(W_n)_n$ on $\Z$ defined by:
\[
W_{n+1} = \Phi_2(W_n, \theta_n) \:.
\]
(The process $(W_n)_n$ is different from the process
$(X_2(n)-Y_2(n))_n$. The former is always defined according to the
recursion \eref{eq-mc} while the latter is defined according to
\eref{eq-mc} only for the $n$'s such that $X_2(n)+X_3(n)>0$. The
former is Markovian while the latter is not.) 

Condition  $c_1 < c_{-1}$ becomes $\mu_C(2) < \mu_S(1') + \mu_S(2')$ and 
$a_{-1} < a_{1}$ becomes $\mu_S(2') < \mu_C(1) + \mu_C(2)$. 
Both conditions follow from {\sc NCond}. So the auxiliary chain
$(W_n)_n$ is ergodic and its stationary distribution $\pi$
satisfies:
$$
 \pi(0)  =  \Bigl( 1 + \frac{b_{-1}}{a_1 - a_{-1}} + \frac{b_1}{c_{-1} - c_1}  \Bigr)^{-1},\quad
 \pi(\Zp)  =  \pi(0)\frac{b_1}{c_{-1} - c_1},\quad
 \pi(\Zm)  =  \pi(0)\frac{b_{-1}}{a_1 - a_{-1}} \:,
$$
where $\Zp = \{1, 2, \ldots\}$ is the set of strictly positive integers
and $\Zm = \{-1, -2, \ldots\}$ is the set of strictly negative integers. 
{From} now on, we fix an initial condition $W_0$ satisfying 
\[
W_0 \sim \pi, \qquad W_0 \ \indep \ (\theta_n)_n \:.
\]

Let us switch back to the Markov chain $(X,Y)$. 
Set 
\[
L(n) =X_2(n) + X_3(n), \qquad \Delta L(n) = L(n+1) -
L(n)\:.
\]
If $L(n)>0$, then we check by direct inspection that:
\[
\Delta L(n) = \Psi ( X_2(n)-Y_2(n), \theta_n)\:,
\]
where $\Psi$ is a deterministic function. 
We have in particular
\begin{eqnarray*}
\alpha & \stackrel{\mbox{def}}{=} & \expct{\Delta L(n) \; | \; L(n) > 0, Y_2(n) >0} \ = \ \mu(3,2') + \mu(3,3') - \mu(1,1') - \mu(2, 1'),\\
\beta &  \stackrel{\mbox{def}}{=} & \expct{\Delta L(n) \; | \; L(n) >
  0, X_2(n) = Y_2(n) = 0} \ = \ \mu(2, 3') + \mu(3,2') + \mu(3,3') -
\mu(1,1'), \\ 
\gamma & \stackrel{\mbox{def}}{=} & \expct{\Delta L(n) \; | \; L(n) > 0, X_2(n) >0} \ = \ \mu(2, 3') + \mu(3,3') - \mu(1,1') - \mu(1,2').
\end{eqnarray*}

Let us turn again to the auxiliary chain $(W_n)_n$. By performing
the computation, we get
\[
\expct{\Psi (W_n, \theta_n)}= \pi (\Zm) \alpha + \pi (0) \beta + \pi (\Zp)
\gamma \:.
\]
The Ergodic Theorem for Markov Chains, see for instance
\cite[\S 3.4]{brem99}, gives:
$$\lim_{n} \frac{1}{n} \sum_{i=0}^{n-1} \Psi (W_n, \theta_n) = 
\pi (\Zm) \alpha + \pi (0) \beta + \pi (\Zp) \gamma \quad
\textrm{a.s.}
$$
Assume that  $\pi (\Zm) \alpha + \pi (0) \beta + \pi (\Zp) \gamma
> 0$. 
Then we have:  
\begin{equation}
\label{eq:phi1}
 \lim_{n} \sum_{i=0}^{n-1} \Psi (W_n, \theta_n) = 
+\infty \quad \textrm{a.s.}
\end{equation}
Therefore, for each $\varepsilon > 0$, there exists $K_{\varepsilon} \geq 0$ such that
\begin{equation}\label{eq-arg}
P \Bigl(\min_{n \geq 1} \sum_{i=0}^{n-1} \Psi (W_n, \theta_n)  > -K_{\varepsilon} \Bigr) 
\geq 1-\varepsilon \:.
\end{equation}

Let us switch back to the Markov chain $(X(n),Y(n))_n$.
Choose the initial condition $(X(0), Y(0))$ such that
\begin{equation}\label{eq-ic}
X_2(0) - Y_2(0) =W_0 \qquad \min (X_3(0), Y_3(0)) = K_{\varepsilon} \:,
\end{equation}
where $K_{\varepsilon}$ is defined in \eref{eq-arg}. 
By construction, on the event $\cA= \{\min_{n \geq 1} \sum_{i=0}^{n-1} \Psi
(W_n, \theta_n)  > -K_{\varepsilon} \}$, we have 
\[
\forall n, \qquad L(n)>0, \quad X_2(n)-Y_2(n)= W_n\:.
\]
So, on the event $\cA$, we have
\begin{equation}
\label{eq:phi2}
L(n) = K_{\varepsilon} + \sum_{i=0}^{n-1} \Delta L(i) = K_{\varepsilon} + \sum_{i=0}^{n-1}
\Psi (W_i, \theta_i) \ \longrightarrow +\infty \:.
\end{equation}
We conclude that the Markov chain $(X,Y)$ of the NN-model is transient.

\medskip

We now show that the stability region is not maximal, by giving an example 
such that $\pi (\Zm) \alpha + \pi (0) \beta + \pi (\Zp) \gamma >
0$. 
Consider $\mu_C = \left ( 1/3 , 2/5, 4/15  \right)$, 
$\mu_S = \mu_C$, and
$\mu = \mu_C \times \mu_S$.
Thus conditions {\sc NCond} are satisfied. However, we have:
\begin{align*}
a_{1} = \frac{11}{25}, \quad a_0 = \frac{34}{75}, \quad  a_{-1} = \frac{8}{75},\\
b_{1} = \frac{6}{25},  \quad b_0 = \frac{13}{25}, \quad  b_{-1} = \frac{6}{25},\\
c_{1} = \frac{8}{75},  \quad c_0 = \frac{34}{75}, \quad  c_{-1} = \frac{11}{25},
\end{align*}
and
$$
\pi(0) = \frac{25}{61}, \quad \pi(\Zp) = \frac{18}{61}, \quad \pi(\Zm) = \frac{18}{61}.
$$
This gives $\alpha = -1/15, \ \beta = 13/75, \ \gamma = -1/15$, 
and,
\[
\pi (\Zm) \alpha + \pi (0) \beta + \pi (\Zp)
\gamma = \frac{29}{915} > 0\:.
\]
This completes the proof. 

\medskip

Consider now the MS policy. Set $L(n) = \min ( X_3(n) - X_2(n), \;
Y_3(n) - Y_2(n) )$. 
The initial distribution can be taken such that 
\[
X_2(0) - Y_2(0) \sim
 \pi, \qquad \begin{cases} X_3(0) - X_2(0) = K_{\varepsilon} & \mbox{if } X_2(0) > 0
   \\ Y_3(0) - Y_2(0) = K_{\varepsilon} &  \mbox{otherwise} \end{cases} \:.
\]
Modulo these modifications, the proof carries over unchanged. 
\endproof

\section{ML is always stable}\label{se-ml}

In this section, we show that the ML policy
has a maximal stability region. 

\medskip

The idea of the proof is as follows. 
Consider the quadratic Lyapunov function:
\begin{equation}\label{eq-lyap}
L(x,y) =  \sum_{c \in C}x_c^2 + \sum_{s \in S} y_s^2, \quad (x,y)\in \CE. 
\end{equation} 
Observe that the
ML policy minimizes the value of this Lyapunov function at each
step. 
We introduce an alternate policy that depends on the arrival
distribution $\mu$. For this policy, we manage to prove that the quadratic Lyapunov
function has a negative drift outside a finite region.

\begin{theorem} For any bipartite graph, the ML policy
has a maximal stability region. 
\end{theorem}

\begin{proof}
We introduce an alternate matching policy. This policy is admissible,
corresponds to a commutative state space,
but does not belong
to the policies listed in \S \ref{sse-match}. It is a random policy and its specificity is to be 
facet dependent.

Let us describe the alternate policy
on a non-empty facet $\cF$. Set $\Cb = \Cb(\cF), \Sb = \Sb(\cF), \Cz=\Cz(\cF),$ etc. 
To describe the matching policy, the only thing we have to describe is
the way to match an arriving customer of class $c\in \Cz$,
resp. server of class $s\in \Sz$. 
Let us concentrate first on a server of class $s\in \Sz$. 

\medskip

{From} {\sc NCond}:
$$\mu_C(\Cb) < \mu_S(\Sz), \qquad  \mu_S(\Sb)<\mu_C(\Cz).$$ 
We build a directed graph as in \eref{eq-dg} but restricted to the
nodes in $\Cb$ and $\Sz$. Formally, 
\begin{equation}\label{eq-dg2}
\cN_{\cF} = \bigl(\Cb\cup \Sz \cup \{i,f\},\{ E \cap
\Cb\times \Sz\} \cup\{(i,c), c\in
\Cb \} \cup \{(s,f), s\in \Sz \}\bigr)\:. 
\end{equation}
Endow the arcs of $E\cap
\Cb \times \Sz$ with infinite capacity, an arc of type $(i,c)$
with capacity $\mu_C(c)$, and an arc of type $(s,f)$ with capacity
$\mu_S(s)$. 

\medskip

As in Lemma \ref{le-flow}, {\sc NCond} implies that the minimal cut of $\cN_{\cF}$ has capacity 
$\mu_C(\Cb)$. Any maximal flow $T´$ is such that: $\forall c\in
\Cb,\ T´(i,c)=\mu_C(c)$, $\forall s \in
\Sz, \ T´(s,f) \leq \mu_S(s)$. Let us prove that there exists 
a maximal flow $T$ such that:
\begin{eqnarray*}
\forall c\in
\Cb, && T(i,c)=\mu_C(c)  \\
\forall s \in
\Sz, && T(s,f) < \mu_S(s)  \:. 
\end{eqnarray*}
Define $\widetilde{\mu}_S$ on $\Sz$ by $\widetilde{\mu}_S(s) = \mu_S(s) - \eta$. Here $\eta >0$ is chosen to be small enough so that: 
$\forall U\subset \Cb, \mu_C(U) < \widetilde{\mu}_S(S(U)), \ \forall V\subset \Sz, \widetilde{\mu}_S(V) < \mu_C(C(V))$. This is possible since {\sc NCond} are open conditions. Consider 
the same network as above but with the capacities $\widetilde{\mu}_S(s )$ on the arcs $(s,f)$. 
The minimal cut still has capacity $\mu_C(\Cb)$. A maximal flow $T$ is such that: 
$\forall c\in
\Cb,\ T(i,c)=\mu_C(c)$, $\forall s \in
\Sz, \ T(s,f) \leq \widetilde{\mu}_S(s)< \mu_S(s)$. Clearly, $T$ is also a flow for the original network. 

\medskip

The server $s\in \Sz$ is matched to $c\in \Cb \cap C(s)$ randomly, independently of
the past, with probability:
\[
P^{\cF}_{sc}= \frac{1}{\mu_S(s)} \Bigl[ \ T(c,s) +
  \frac{\mu_S(s)-T(s,f)}{|\Cb \cap C(s)|} \ \Bigr] \:.
\]
Let us check that this defines indeed a probability:
\begin{eqnarray*}
\sum_{c\in \Cb \cap C(s)} P^{\cF}_{sc} & = &  \frac{1}{\mu_S(s)} \Bigl[ \ 
  \mu_S(s)-T(s,f) + \sum_{c\in \Cb \cap C(s)}
  T(c,s) \ \Bigr] \\
& = &  \frac{1}{\mu_S(s)} \Bigl[ \ \mu_S(s)-T(s,f)  + T(s,f) \ \Bigr]
\ \  =
  \ \  1 \:.
\end{eqnarray*}
For $c\in \Cb$, $s\in S(c)$, set $\varepsilon_{sc}=
(\mu_S(s)-T(s,f))/|\Cb \cap C(s)|$. For  $c\in \Cb$, set
 $\varepsilon_{c}= \sum_{s\in S(c)} \varepsilon_{sc}$. We have $\varepsilon_{c}>0$. 
Observe that:
\begin{equation}
 \label{eq:vn}
 \forall c\in \Cb, \qquad \sum_{s\in S(c)} \mu_S(s)P^{\cF}_{sc} =
 \mu_C(c) + \varepsilon_c \:.
\end{equation}

Symmetrically, we define the directed graph of type \eref{eq-dg2} but
on the nodes $\Cz$ and $\Sb$. We build a maximal flow on
this new graph as above, and based on this flow, we define the
probability $P^{\cF}_{cs}$ that a customer $c\in
\Cz$ is matched to a server $s\in \Sb \cap S(c)$. For $s\in \Sb$, 
we define $\varepsilon_{cs}, \; c\in C(s)$, and $\varepsilon_{s}$
accordingly. We have $\varepsilon_{s}>0$. 

\medskip

Let $(X(n), Y(n))_{n}$ be the Markov chain of the buffer-content of the model. 
Assume that $(X(n), Y(n))= (x,y) \in \cF$ and let $c\in \Cb$. We have:
\begin{itemi}
\item[(i)] $X(n+1)_c = X(n)_c -1$ iff: 
\begin{itemi} 
\item[-] the arriving customer is not of class $c$;
\item[-] the arriving server is of class $s\in S(c)$; 
\item[-] the arriving server is matched with $c$ (probability
$P^{\cF}_{sc}$). 
\end{itemi}
This case happens with probability $\alpha_c = \sum_{s \in S(c)} \mu(C-c, s)P^{\cF}_{sc}$.
\item[(ii)] $X(n+1)_c = X(n)_c +1$ iff:  
\begin{itemi}
\item[-] the arriving customer is of class $c$;
\item[-] the arriving server is not matched with $c$. This may occur in two possible ways: either the
arriving server is of class $s\not\in S(c)$, or the arriving server is
of class $s\in S(c)$ but is not matched with $c$ (probability
$1-P^{\cF}_{sc}$).
\end{itemi}

This case happens with probability $\beta_c = \mu(c,S-S(c)) + \sum_{s \in S(c)} \mu(c, s)(1-P^{\cF}_{sc})$. 
\item[(iii)] Otherwise, $X(n+1)_c = X(n)_c$. 
\end{itemi}

Using (\ref{eq:vn}), we get:
\begin{eqnarray*}
\sum_{s \in S(c)} \mu(C-c, s)P^{\cF}_{sc} & = & \sum_{s\in S(c)} \mu_S(s)P^{\cF}_{sc} - \sum_{s \in S(c)} \mu(c, s)P^{\cF}_{sc} \\
& = & \mu_C(c) + \varepsilon_c - \sum_{s \in S(c)} \mu(c, s)P^{\cF}_{sc} \\
& = & \mu(c, S- S(c)) + \mu(c, S(c)) + \varepsilon_c - \sum_{s \in S(c)} \mu(c, s)P^{\cF}_{sc} \\
& = & \mu(c,S-S(c)) + \sum_{s \in S(c)} \mu(c, s)(1-P^{\cF}_{sc}) + \varepsilon_c.
\end{eqnarray*}
Thus, $\alpha_c = \beta_c + \varepsilon_c$. Observe that $\beta_c < \mu_C(c)$.
We get, for $(x, y) \in \cF$ and $c \in \Cb$:
\begin{align*}
\expct{X(n+1)_c^2 - X(n)_c^2 \; |\; (X(n), Y(n))= (x,y)} &= \beta_c (2 x_c +1) - \alpha_c(2 x_c -1) \\
&= 2\beta_c - \varepsilon_c (2 x_c -1)\\
& < 2\mu_C(c)- \varepsilon_c x_c.
\end{align*}

Let $L$ be the quadratic Lyapunov function \eref{eq-lyap}. 
Define $\Delta L(n) = L(X_{n+1}, Y_{n+1}) - L(X_{n}, Y_{n})$.
Set $\varepsilon = \min_{v \in \Cb \cup \Sb} \varepsilon_v >0$.
Then (the first term in the sum takes care of the vertices in $C - \Cb$ and $S - \Sb$):
\begin{align*}
\expct{\Delta L(n) \; | \; (X_{n}, Y_{n}) = (x,y)} & <  
2 + \sum_{c \in \Cb} \left (2\mu_C(c) - \varepsilon_c x_c  \right )
+ \sum_{s \in \Sb} \left ( 2 \mu_S(s) - \varepsilon_s y_s  \right )\\
& < 2 + 2 \mu_C(\Cb) + 2 \mu_S(\Sb) - \varepsilon \left(\sum_{c \in \Cb}x_c+\sum_{s \in \Sb}y_s \right)  \\
&< 6 - 2\:\varepsilon \sum_{c \in C}x_c \:.
\end{align*}
Fix $\delta > 0$. If $ \sum_{c \in C}x_c > (6+\delta)/2 \varepsilon$,  
then $\expct{\Delta L(n)} < -\delta$. 
There are finitely many facets, so there is a finite set $A\subset \CE$ such that 
\begin{equation}\label{eq-forlyap}
\forall (x,y) \not\in A, \quad \expct{\Delta L(n)} < -\delta \:. 
\end{equation}
By the Lyapunov-Foster's Theorem, see for instance \cite[\S 5.1]{brem99}, the alternate matching policy is stable.  

Since the ML matching policy minimizes the value of the quadratic Lyapunov function, we have a fortiori 
that \eref{eq-forlyap} holds for it. Therefore, the ML policy
is also stable. 
\end{proof}

\paragraph{Conclusion.}

Many open questions remain. First, we do not know if the stability
region is always maximal for the FIFO and Random policies. Numerical experiments
seem to indicate that it is indeed the case. Second, for the MS and
priority policies, we know that the stability region is not always
maximal, but we do not know how to compute it. Last,
we would like 
to obtain sufficient conditions for stability, valid for all
admissible policies, and which are better than
the ones of \S \ref{se-suff}. 

The program used to carry out the numerical experiments is available
on request from Ana Bu\v si\'c.

\end{document}